\documentclass[preprint]{elsarticle}
\usepackage{microtype}
\usepackage{stmaryrd} 
\usepackage{tikz}
\usepackage{array}
\usepackage{mathrsfs}
\usepackage{amsmath,amsfonts,amssymb} 
\usepackage{amsthm}
\usepackage{url}
\usepackage{stackrel}
\usepackage{subcaption}
\usepackage{listings}
\usetikzlibrary{decorations.pathmorphing,calc,shapes.geometric}

\newtheorem{definition}{Definition}
\newtheorem{example}{Example}
\newtheorem{theorem}{Theorem}
\newtheorem{corollary}{Corollary}
\newtheorem{lemma}{Lemma}
\newtheorem{proposition}{Proposition}
\newtheorem{remark}{Remark}

\newdimen\proofrulebreadth \proofrulebreadth=.05em
\newdimen\proofdotseparation \proofdotseparation=1.25ex
\newdimen\proofrulebaseline \proofrulebaseline=2ex
\newcount\proofdotnumber \proofdotnumber=3
\let\then\relax
\def\hfi{\hskip0pt plus.0001fil}
\mathchardef\squigto="3A3B
\newif\ifinsideprooftree\insideprooftreefalse
\newif\ifonleftofproofrule\onleftofproofrulefalse
\newif\ifproofdots\proofdotsfalse
\newif\ifdoubleproof\doubleprooffalse
\let\wereinproofbit\relax
\newdimen\shortenproofleft
\newdimen\shortenproofright
\newdimen\proofbelowshift
\newbox\proofabove
\newbox\proofbelow
\newbox\proofrulename
\def\shiftproofbelow{\let\next\relax\afterassignment\setshiftproofbelow\dimen0 }
\def\shiftproofbelowneg{\def\next{\multiply\dimen0 by-1 }%
\afterassignment\setshiftproofbelow\dimen0 }
\def\setshiftproofbelow{\next\proofbelowshift=\dimen0 }
\def\setproofrulebreadth{\proofrulebreadth}

\def\prooftree{
\ifnum  \lastpenalty=1
\then   \unpenalty
\else   \onleftofproofrulefalse
\fi
\ifonleftofproofrule
\else   \ifinsideprooftree
        \then   \hskip.5em plus1fil
        \fi
\fi
\bgroup
\setbox\proofbelow=\hbox{}\setbox\proofrulename=\hbox{}%
\let\justifies\proofover\let\leadsto\proofoverdots\let\Justifies\proofoverdbl
\let\using\proofusing\let\[\prooftree
\ifinsideprooftree\let\]\endprooftree\fi
\proofdotsfalse\doubleprooffalse
\let\thickness\setproofrulebreadth
\let\shiftright\shiftproofbelow \let\shift\shiftproofbelow
\let\shiftleft\shiftproofbelowneg
\let\ifwasinsideprooftree\ifinsideprooftree
\insideprooftreetrue
\setbox\proofabove=\hbox\bgroup$\displaystyle 
\let\wereinproofbit\prooftree
\shortenproofleft=0pt \shortenproofright=0pt \proofbelowshift=0pt
\onleftofproofruletrue\penalty1
}

\def\eproofbit{
\ifx    \wereinproofbit\prooftree
\then   \ifcase \lastpenalty
        \then   \shortenproofright=0pt  
        \or     \unpenalty\hfil         
        \or     \unpenalty\unskip       
        \else   \shortenproofright=0pt  
        \fi
\fi
\global\dimen0=\shortenproofleft
\global\dimen1=\shortenproofright
\global\dimen2=\proofrulebreadth
\global\dimen3=\proofbelowshift
\global\dimen4=\proofdotseparation
\global\count255=\proofdotnumber
$\egroup  
\shortenproofleft=\dimen0
\shortenproofright=\dimen1
\proofrulebreadth=\dimen2
\proofbelowshift=\dimen3
\proofdotseparation=\dimen4
\proofdotnumber=\count255
}

\def\proofover{
\eproofbit 
\setbox\proofbelow=\hbox\bgroup 
\let\wereinproofbit\proofover
$\displaystyle
}%
\def\proofoverdbl{
\eproofbit 
\doubleprooftrue
\setbox\proofbelow=\hbox\bgroup 
\let\wereinproofbit\proofoverdbl
$\displaystyle
}%
\def\proofoverdots{
\eproofbit 
\proofdotstrue
\setbox\proofbelow=\hbox\bgroup 
\let\wereinproofbit\proofoverdots
$\displaystyle
}%
\def\proofusing{
\eproofbit 
\setbox\proofrulename=\hbox\bgroup 
\let\wereinproofbit\proofusing
\kern0.3em$
}

\def\endprooftree{
\eproofbit 
  \dimen5 =0pt
\dimen0=\wd\proofabove \advance\dimen0-\shortenproofleft
\advance\dimen0-\shortenproofright
\dimen1=.5\dimen0 \advance\dimen1-.5\wd\proofbelow
\dimen4=\dimen1
\advance\dimen1\proofbelowshift \advance\dimen4-\proofbelowshift
\ifdim  \dimen1<0pt
\then   \advance\shortenproofleft\dimen1
        \advance\dimen0-\dimen1
        \dimen1=0pt
        \ifdim  \shortenproofleft<0pt
        \then   \setbox\proofabove=\hbox{%
                        \kern-\shortenproofleft\unhbox\proofabove}%
                \shortenproofleft=0pt
        \fi
\fi
\ifdim  \dimen4<0pt
\then   \advance\shortenproofright\dimen4
        \advance\dimen0-\dimen4
        \dimen4=0pt
\fi
\ifdim  \shortenproofright<\wd\proofrulename
\then   \shortenproofright=\wd\proofrulename
\fi
\dimen2=\shortenproofleft \advance\dimen2 by\dimen1
\dimen3=\shortenproofright\advance\dimen3 by\dimen4
\ifproofdots
\then
        \dimen6=\shortenproofleft \advance\dimen6 .5\dimen0
        \setbox1=\vbox to\proofdotseparation{\vss\hbox{$\cdot$}\vss}%
        \setbox0=\hbox{%
                \advance\dimen6-.5\wd1
                \kern\dimen6
                $\vcenter to\proofdotnumber\proofdotseparation
                        {\leaders\box1\vfill}$%
                \unhbox\proofrulename}%
\else   \dimen6=\fontdimen22\the\textfont2 
        \dimen7=\dimen6
        \advance\dimen6by.5\proofrulebreadth
        \advance\dimen7by-.5\proofrulebreadth
        \setbox0=\hbox{%
                \kern\shortenproofleft
                \ifdoubleproof
                \then   \hbox to\dimen0{%
                        $\mathsurround0pt\mathord=\mkern-6mu%
                        \cleaders\hbox{$\mkern-2mu=\mkern-2mu$}\hfill
                        \mkern-6mu\mathord=$}%
                \else   \vrule height\dimen6 depth-\dimen7 width\dimen0
                \fi
                \unhbox\proofrulename}%
        \ht0=\dimen6 \dp0=-\dimen7
\fi
\let\doll\relax
\ifwasinsideprooftree
\then   \let\VBOX\vbox
\else   \ifmmode\else$\let\doll=$\fi
        \let\VBOX\vcenter
\fi
\VBOX   {\baselineskip\proofrulebaseline \lineskip.2ex
        \expandafter\lineskiplimit\ifproofdots0ex\else-0.6ex\fi
        \hbox   spread\dimen5   {\hfi\unhbox\proofabove\hfi}%
        \hbox{\box0}%
        \hbox   {\kern\dimen2 \box\proofbelow}}\doll%
\global\dimen2=\dimen2
\global\dimen3=\dimen3
\egroup 
\ifonleftofproofrule
\then   \shortenproofleft=\dimen2
\fi
\shortenproofright=\dimen3
\onleftofproofrulefalse
\ifinsideprooftree
\then   \hskip.5em plus 1fil \penalty2
\fi
}

\newcommand{\ninfer}[3]
     {\prooftree
          #1 
          \justifies #2
          \using #3
      \endprooftree}

\newcommand{\FV}{\ens{V}}
\newcommand{\ens}[1]{\mathbb{#1}}
\newcommand{\dom}{\textit{dom}}
\newcommand{\ea}{\textnormal{\tt e}}
\newcommand{\mea}{\textnormal{\tt me}}
\newcommand{\eaa}{\ea_1}
\newcommand{\eai}{\ea_i}
\newcommand{\ean}{\ea_n}

\newcommand{\Var}{\FV}
\newcommand{\xa}{\variable{x}}
\newcommand{\variable}[1]{\textnormal{\tt #1}}
\newcommand{\xaa}{\variable{x}_1}
\newcommand{\xai}{\variable{x}_i}
\newcommand{\xan}{\variable{x}_n}
\newcommand{\xb}{\variable{y}}
\newcommand{\xc}{\variable{z}}
\newcommand{\xd}{\variable{u}}
\newcommand{\vetat}{\variable{state}}

\newcommand{\true}{\textnormal{\tt true}}
\newcommand{\false}{\textnormal{\tt false}}

\newcommand{\Operator}{\ens{O}}

\newcommand{\Class}{\ens{C}}
\newcommand{\Method}{\ens{M}}
\newcommand{\nest}{\nu}
\newcommand{\lev}{\lambda}
\newcommand{\oper}[1]{\tt{#1}}
\newcommand{\op}{\oper{op}} 
\newcommand{\new}{\textnormal{\tt new}}
\newcommand{\nul}{\textnormal{\tt null}}
\newcommand{\extends}{\textnormal{\tt extends}}
\newcommand{\charac}{\textnormal{\tt char}}
\newcommand{\bool}{\textnormal{\tt boolean}}

\newcommand{\void}{\textnormal{\tt void}}
\newcommand{\ent}{\textnormal{\tt int}}
\newcommand{\return}{\textnormal{\tt return}}
\newcommand{\breaks}{\textnormal{\tt break}}
\newcommand{\main}{\textnormal{\tt main}}
\newcommand{\this}{\textnormal{\tt this}}
\newcommand{\pop}{\textnormal{\tt pop}}
\newcommand{\push}{\textnormal{\tt push}}

\newcommand{\Command}[1]{\textnormal{\tt #1}}
\newcommand{\ca}{\Command{I}}
\newcommand{\mi}{\Command{MI}}
\newcommand{\mj}{\Command{MJ}}
\newcommand{\caa}{\ca_{\textnormal{\tt 1}}}
\newcommand{\cab}{\ca_{\textnormal{\tt 2}}}
\newcommand{\cb}{\Command{J}}
\newcommand{\cc}{\Command{K}}
\newcommand{\cd}{\Command{L}}
\newcommand{\instr}[1]{\mathtt{#1}}
\newcommand{\iasg}{\instr{:=}}
\newcommand{\iwh}{\instr{while}}
\newcommand{\iwhile}{\instr{while}}
\newcommand{\ifor}{\instr{for}}
\newcommand{\cond}{\instr{condition}}
\newcommand{\ins}{\instr{Ins}}
\newcommand{\inc}{\instr{Increment}}
\newcommand{\iret}{\instr{return}}
\newcommand{\iif}{\instr{if}}
\newcommand{\ithen}{\instr{}}
\newcommand{\ielse}{\instr{else}}
\newcommand{\ido}{\instr{}}

\newcommand{\m}{\textnormal{\tt m}}
\newcommand{\n}{\textnormal{\tt cst}_\tau}

\newcommand{\C}{\textnormal{\tt C}}
\newcommand{\D}{\textnormal{\tt D}}
\newcommand{\Cons}{\textnormal{\tt k}}
\newcommand{\class}{\mathcal{C}}
\newcommand{\Exe}{\textnormal{\tt Exe}}

\newcommand{\pom}{p_\heap}
\newcommand{\stack}{\mathcal{S}}
\newcommand{\heap}{\mathcal{H}}
\newcommand{\conf}{\mathcal{C}}
\newcommand{\conff}{\mathcal{D}}
\newcommand{\data}{\mathcal{I}}

\newcommand{\rgl}{::=}
\newcommand{\typenv}{\Gamma}

\newcommand{\Imp}{\vDash}
\newcommand{\imp}{\vdash}
\newcommand{\sem}[1]{\llbracket #1 \rrbracket}

\newcommand{\slat}[1]{\ty{#1}}
\newcommand{\sla}{\slat{\alpha}}

\newcommand{\ord}{\preceq}
\newcommand{\join}{\vee}
\newcommand{\tier}[1]{\mathbf{#1}}
\newcommand{\tiera}{\tier{0}}
\newcommand{\tierb}{\tier{1}}
\newcommand{\ty}[1]{#1}
\newcommand{\tya}{\ty{\alpha}}

\newcommand{\tyb}{\ty{\beta}}

\newcommand{\dord}{\unlhd} 
\newcommand\mcall{\sqsubset}
\DeclareFontFamily{U}{mathb}{\hyphenchar\font45}
\DeclareFontShape{U}{mathb}{m}{n}{
      <5> <6> <7> <8> <9> <10> gen * mathb
      <10.95> mathb10 <12> <14.4> <17.28> <20.74> <24.88> mathb12
      }{}
\DeclareSymbolFont{mathb}{U}{mathb}{m}{n}
\DeclareFontSubstitution{U}{mathb}{m}{n}
\DeclareMathSymbol{\sqsubsetneqq}   {3}{mathb}{"90}
\newcommand\smcall{\sqsubsetneqq}

\newcommand{\taille}[1]{|#1|}
\newcommand{\size}[1]{|#1|}

\title{A Type-Based Complexity Analysis of Object Oriented Programs}
\author{Emmanuel Hainry}
\ead{hainry@loria.fr}
\author{Romain P\'echoux}
\ead{pechoux@loria.fr}
\address{Universit{\'e} de Lorraine, CNRS, Inria, LORIA, Nancy, France}
\fntext[foot1]{This work has been partially supported by ANR Project ELICA ANR-14-CE25-0005.}

\begin{document}
\begin{abstract}
A type system is introduced for a generic Object Oriented programming language in order to infer resource upper bounds. A sound and complete characterization of the set of polynomial time computable functions is obtained. As a consequence, the heap-space and the stack-space requirements of typed programs are also bounded polynomially. This type system is inspired by previous works on Implicit Computational Complexity, using tiering and non-interference techniques. The presented methodology has several advantages. First, it provides explicit big $O$ polynomial upper bounds to the programmer, hence its use could allow the programmer to avoid memory errors.  Second, type checking is decidable in polynomial time. Last, it has a good expressivity since it analyzes most object oriented features like inheritance, overload, override and recursion. Moreover it can deal with loops guarded by objects and can also be extended to statements that alter the control flow like break or return.
\end{abstract}

\maketitle

\section{Introduction}

\subsection{Motivations}
In the last decade, the development of embedded systems and mobile computing has led to a renewal of interest in predicting program resource consumption. This kind of problematic is highly challenging for popular object oriented programming languages which come equipped with environments for applications running on mobile and other embedded devices (e.g. Dalvik, Java Platform Micro Edition (Java ME), Java Card and Oracle Java ME Embedded). 

The current paper tackles this issue by introducing a type system for a compile-time analysis of both heap and stack space requirements of OO programs thus avoiding memory errors. This type system is also sound and complete for the set of polynomial time computable functions on the Object Oriented paradigm. 

This type system combines ideas coming from tiering discipline, used for complexity analysis of function algebra~\cite{BC92,LM93}, together with ideas coming from non-interference, used for secure information flow analysis~\cite{VIS96}. The current work is an extended version of~\cite{HP15} and is strongly inspired by the seminal paper~\cite{M11}.

\subsection{Abstract OO language}
The results of this paper will be presented in a formally oriented manner in order to highlight their theoretical soundness. For this, we will consider a generic  Abstract Object Oriented language called AOO. It can be seen as a language strictly more expressive than Featherweight Java~\cite{IPW01} enriched with features like variable updates and while loops. The language is generic enough. Consequently, the obtained results can be applied both to impure OO languages (\emph{e.g.} Java) and to pure ones (\emph{e.g}. SmallTalk or Ruby). Indeed, in this latter case, it just suffices to forget rules about primitive data types in the type system. Moreover, it does not depend on the implementation of the language being compiled (ObjectiveC, OCaml, Scala, ...) or interpreted (Python standard implementation, OCaml, ...). There are some restrictions: it does not handle exceptions, inner classes, generics, multiple inheritance or pointers. Hence languages such as C++ cannot be handled. However we claim that the analysis can be extended to exceptions, inner classes and generics. This is not done in the paper in order to simplify the technical analysis. The presented work captures Safe Recursion on Notation by Bellantoni and Cook~\cite{BC92} and we conjecture that it could be adapted to programs with higher-order functions. The intuition behind such a conjecture is just that tiers are very closely related to the $!$ and $\S$ modalities of light logics~\cite{Girard98}.

\subsection{Intuition}
The heap is represented by a directed graph where nodes are object addresses and arrows relate an object address to its field addresses. The type system splits variables in two universes: tier $\tiera$ universe and tier $\tierb$ universe. In this setting, the high security level is tier $\tiera$ while low security level is tier $\tierb$. While tier $\tierb$ variables are pointers to nodes of the initial heap, tier $\tiera$ variables may point to newly created addresses. The information may flow from tier $\tierb$ to tier $\tiera$, that is a tier $\tiera$ variable may depend on tier $\tierb$ variables. However the presented type system precludes flows from $\tiera$ to $\tierb$. Indeed once a variable has stored a newly created instance, it can only be of tier $\tiera$. Tier $\tierb$ variables are the ones that can be used either as guards of a while loop or as a recursive argument in a method call whereas tier $\tiera$ variables are just used as storages for the computed data. This is the reason why, in analogy with information-flow analysis, tier $\tiera$ is the high security level of the current setting, though this naming is opposed to the {\sc icc} standard interpretation where tier $\tierb$ is usually seen as ``safer'' than $\tiera$ because its use is controlled and restricted.

The polynomial upper bound is obtained as follows: if the input graph structure has size $n$ then the number of distinct possible configurations for $k$ tier $\tierb$ variables is at most $O(n^k)$. For this, we put some restrictions on operations that can make the memory grow : constructors for the heap and operators and method calls for the stack.

\subsection{Example}
Consider the following Java code duplicating the length of a boolean $\tt BList$ as an illustrating example:
\vspace*{-0.2cm}
\begin{verbatim}
y := x.clone();
while (x != null){
    y := new BList(true,y);
    x := x.getQueue();
}
\end{verbatim}
\vspace*{-0.2cm}
The tier of variable $\verb!x!$ will be enforced to be $\tierb$ since it is used in a while loop guard and in the call of the recursive method $\verb!clone!$. On the opposite, the tier of variable $\verb!y!$ will be enforced to be $\tiera$ since the $\verb!y:=new BList(true,y);!$ instruction enlarges the memory use. For each assignment, we will check that the tier of the variable assigned to is equal to (smaller than for primitive data) the tier of the assigned expression. Consequently, the assignment $\verb!y:=x.clone();!$ is typable in this code (since the call $\verb!x.clone();!$ is of tier $\tiera$ as it makes the memory grow) whereas it cannot be typed if the first instruction is to be replaced by either $\verb!x:=y.clone();!$ or $\xa \iasg \xb$. 

\subsection{Methodology}
The OO program complexity analysis presented in this paper can be summed up by the following figure: 
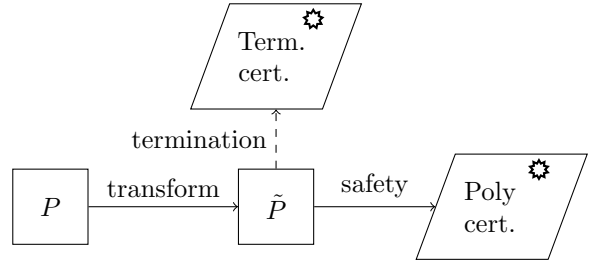
\begin{center}
\begin{tikzpicture}[node distance=3cm,
fleche/.style={->},
para/.style={draw, trapezium, trapezium left angle=60, trapezium right angle=120, trapezium stretches=true, text width=1cm, minimum height=14mm},
box/.style={draw,rounded corners,text width=11mm},
carre/.style={draw,minimum size=1cm}]
\node[carre] (p) {$P$};
\node[carre] (q) [right of=p] {$\tilde{P}$};
\node[para] (c1) [right of=q] {Poly cert.};
\node[draw,thick, star, star points=10, node distance=14mm,scale=.5] (s1) [above right of=c1] {}; 
\node[para,node distance=2cm] (c2) [above of=q] {Term. cert.};
\node[draw,thick, star, star points=9, node distance=14mm,scale=.5] (s2) [above right of=c2] {}; 

\draw[fleche] (p) to node[above] {transform}
(q);
\draw[fleche] (q) to node[above] {safety} (c1);
\draw[fleche, dashed] (q) to node[left] {termination} (c2);
\end{tikzpicture}
\end{center}

In a first step, given a program $P$ of a given OO programming language, we first apply a transformation step in order to obtain the AOO program $\tilde{P}$. This transformation contains the following steps:
\begin{itemize}
\item convert each syntactical construct of the source language in $P$ to the corresponding construct in the abstract OO language. In particular, for statements can be replaced by while statements,
\item for all public fields of $P$, write the corresponding getter and setter in $\tilde{P}$,
\item $\alpha$-convert the variables so that there is no name clashes in $\tilde{P}$,
\item flatten the program (this will be explained in Subsection~\ref{flatflat}).
\end{itemize}
All these steps can be performed in polynomial time and the program abstract semantics is preserved. Consequently, $P$ terminates iff $\tilde{P}$ terminates.

In a second step,  a termination check and a safety check can be performed in parallel. The termination certificate can be obtained using existing tools (see the related works subsection). As the semantics is preserved, the check can also be performed on the original program $P$ or on the compiled bytecode. 
In the safety check, a polynomial time type inference (Proposition~\ref{poltypeinf}) is performed together with a safety criterion check on recursive methods. This latter check 
called \emph{safety} can be checked in polynomial time. A more general criterion, called general safety and which is unfortunately undecidable, is also provided. If both checks succeed, Theorem~\ref{sound} ensures polynomial time termination. 

\subsection{Outline}
In Section 2 and 3, the syntax and semantics of the considered generic language AOO are presented. Note that the semantics is defined on meta-instructions, flattened instructions that make the semantics formal treatment easier. The main contribution: a tier based type system is presented in Section 4 together with illustrating examples. In Section 5, two criteria for recursive methods are provided: the safety that is decidable in polynomial time and the general safety that is strictly more expressive but undecidable. Section 6 establishes the main non-interference properties of the type system. Section 7 and 8 are devoted to prove soundness and, respectively, completeness of the main result: a characterization of the set of polynomial time computable functions. As an aside, explicit space upper bounds on the heap and stack space usage are also obtained. Section 9 is devoted to prove the polynomial time decidability of type inference. Section 10 discusses extensions improving the expressivity, including an extension based on declassification.

\section{Syntax of AOO} \label{syn}
In this section, the syntax of an Abstract Oriented Object programming language, called AOO, is introduced. This language is general enough so that any well-known OO programming language (Java, OCaml, Scala, ...) can be compiled to it under some slight restrictions (not all the features of these languages -- threads, user interface, input/output, ... -- are handled and some program transformations/refinements are needed for a practical application).
\paragraph{Grammar}
Given four disjoint sets $\Var$, $\Operator$, $\Method$ and $\Class$ representing the set of variables, the set of operators, the set of method names and the set of class names, respectively, expressions, instructions, constructors, methods and classes are defined by the grammar of Figure~\ref{fig1},
\begin{figure}[!ht]
\hrulefill
 $$
 \begin{array}{rcl}
\text{Expressions} \ni \ea  &\rgl& \xa \ |\ \n \ |\ \nul \ | \ \this \ | \ \op(\overline{\ea}) \\
 &&  |\  \new\ \C(\overline{\ea}) \ |\ \ea.\m(\overline{\ea})  \\ 
\rule{0mm}{1.5em}
\text{Instructions} \ni  \ca &\rgl& ; \ | \ [\tau]\ \xa \iasg \ea ;  \ | \  \caa \ \cab  \ | \  \iwh(\ea)\ido \{ \ca  \} \\
&&  | \ \iif (\ea)\ithen\{ \ca_1\}\ielse \{\ca_2\} \ |\ \ea.\m(\overline{\ea}); \\
\rule{0mm}{1.5em}
\text{Methods} \ni  \m_{\C} &\rgl&  \tau\ \m(\overline{\tau\ \xa}) \{\ca [\return\ {\xa};] \} \\
\rule{0mm}{1.5em}
\text{Constructors} \ni  \Cons_\C &\rgl&  \C(\overline{\tau\ \xb}) \{ \ca \} \\
\rule{0mm}{1.5em}
\text{Classes} \ni  \C &\rgl&  \C\ [\extends\ \D ]\ \{ \overline{\tau\ \xa;}\ \overline{\Cons_\C} \ \overline{\m_\C} \} 
 \end{array}
$$
\caption{Syntax of AOO classes}
\label{fig1}
\hrulefill
\end{figure}
where $\xa \in \Var$,  $\op \in \Operator$, $\m \in \Method$ and $\C \in \Class$, where
$[e]$ denotes some optional syntactic element $e$ and where $\overline{e}$ denotes a sequence of syntactic elements $e_1,\ldots,e_n$. Also assume a fixed set of discrete primitive types $\mathbb{T}$ to be given, $\textit{e.g.}\ \mathbb{T}=\{\void,\bool, \ent, \charac\}$ or $\mathbb{T}=\emptyset$ in the case of a pure OO language. In what follows, $\{\void,\bool\} \subseteq \mathbb{T}$ will always hold. The $\tau$s are type variables ranging over $\Class \cup \mathbb{T}$. The metavariable $\n$ represents a primitive type constant of type $\tau \in \mathbb{T}$. Let $\dom(\tau)$ be the domain of values of type $\tau$. For example, $\dom(\bool)=\{\true,\false\}$ or $\dom(\ent)=\mathbb{N}$. $\n \in \dom(\tau)$ holds. Finally, define $\dom(\mathbb{T})=\cup_{\tau \in \mathbb{T}}\dom(\tau)$.  Each primitive operator $\op \in \Operator$ has a fixed arity $n$ and comes equipped with a signature of the shape $\op::\tau_1 \times \cdots \times \tau_n \to \tau_{n+1}$ fixed by the language implementation and such that $\tau_1,\ldots,\tau_n \in  \Class \cup \mathbb{T}$ and $\tau_{n+1} \in \mathbb{T}$. That is, operator outputs are of primitive type. An example of such operator will be \verb!==!. Also note, that primitive operators can be both considered to be applied to finite data-types as for Java integers as well as infinite datatype as for Python integers. 

 The AOO syntax does not include a $\ifor$ instruction based on the premise that, as in Java, a for statement $\ifor(\tau\ \xa \iasg \ea; \cond ; \inc)\{ \ins \}$ can be simulated by the statement $\tau\ \xa \iasg \ea;$ $\iwh(\cond)\ \ido \{ \ins \ \inc; \}$. Given a method $\tau\ \m(\tau_1\ \xa_1,\ldots,\tau_n\ \xa_n) \{\ca\ [\return\ {\xa};] \}$ of $\C$, its signature is $\tau\ \m^\C(\tau_1,\ldots,\tau_n)$, the notation $\m^\C$ denoting that $\m$ is declared in $\C$. The signature of a constructor $\Cons_\C$ is  $ \C(\overline{\tau})$. Note that method overload is possible as a method name may appear in several distinct signatures.

\paragraph{Variables toponymy}
In a class $\class=  \C \{ \tau_1\ \xa_1;\ldots;\tau_n \ \xa_n;\ \ \overline{\Cons_\C} \ \overline{\m_\C} \}$, the variables $\xa_i$ are called fields. In a method or constructor $\tau\ \m(\overline{\tau\ \xb}) \{\ca [\return\ {\xa};] \}$, the arguments $\xb_j$ are called parameters. Each variable $\xa$ declared in an assignment of the shape ${\tau}\ \xa \iasg \ea ;$ is called a local variable. Hence, in a given class, a variable is either a field, or a parameter or a local variable. Let $\C.\mathcal{F}=\{\overline{\xa}\}$ to be the set of fields in a class $\C\ \{ \overline{\tau\ \xa;}\ \overline{\Cons_\C}\ \overline{\m_\C}\}$ and $\mathcal{F}=\cup_{\C \in \Class} \C.\mathcal{F} \subseteq \Var$ be the set of all fields of a given program $P$, when $P$ is clear from the context.

\paragraph{No access level}
In an AOO program, the fields of an instance cannot be accessed directly using the ``.'' operator. Getters will be needed. This is based on the implicit assumption that all fields are \texttt{private} since there is no field access in the syntax. On the opposite, methods and classes are all \texttt{public}. This is not a huge restriction for an OO programmer since any field can be accessed and updated in an outer class by writing the corresponding getter and setter. 
 
\paragraph{Inheritance}
Inheritance is allowed by the syntax of AOO programs through the use of the $\extends$ construct. Consequently, override is allowed by the syntax.
 In the case where $\C\ \extends\ \D$, the constructors $\C(\overline{\tau\ \xb}) \{\ca\}$ are constructors initializing both the fields of $\C$ and the fields of $\D$. Inheritance defines a partial order on classes denoted by $\C \dord \D$.

\paragraph{AOO programs} A \emph{program} is a collection of classes together with exactly one class $\Exe\{  \void\ \main()\{ \texttt{Init} \ \texttt{Comp} \} \}$ with $ \texttt{Init}, \texttt{Comp} \in \text{Instructions}$. The method $\main$ of class $\Exe$ is intended to be the entry point of the program. The instruction $ \texttt{Init}$ is called the {\emph{initialization instruction}}. Its purpose is to compute the program input, which is strongly needed in order to define the complexity of an AOO program (if there is no input, all terminating programs are constant time programs). The instruction $\texttt{Comp}$ is called the {\emph{computational instruction}}. The type system presented in this paper will analyze the complexity of this latter instruction. See Subsection~\ref{input} for more explanations about such a choice.

\paragraph{Well-formed programs} \label{subSynSem}
Throughout the paper, only well-formed programs satisfying the following conditions will be considered:
\begin{itemize}
\item Each class name $\C$ appearing in the collection of classes corresponds to exactly one class of name $\C \in \Class$.
\item Each local variable $\xa$ is both declared and initialized exactly once by a $\tau\ \xa := \ea;$ instruction for its first use.
\item A method output type is $\void$ iff it has no $\iret$ statement.
\item Each method signature is unique with respect to its name, class and input types. This implies that it is forbidden to define two signatures of the shape $\tau\ \m^\C(\overline{\tau})$ and $\tau'\ \m^\C(\overline{\tau})$ with $\tau \neq \tau'$.
\end{itemize}

\begin{example}\label{ex1}
Let the class \verb!BList! be an encoding of binary integers as lists of booleans (with the least significant bit in head). The complete code will be given in Example~\ref{ex8}.

{\tt
\begin{lstlisting}[basicstyle=\footnotesize]
  BList {
    boolean value;
    BList queue;

    BList() {
      value := true;
      queue := null;
    }
    
    BList(boolean v, BList q) {
      value := v;
      queue := q;
    }

    BList getQueue() { return queue; }

    void setQueue(BList q) {
      queue := q;
    }

    boolean getValue() { return value; }
    
    ... 
  }
\end{lstlisting}
}
\end{example}

\section{Semantics of AOO}\label{sem}
In this section, a pointer graph semantics of AOO programs is provided. Pointer graphs are reminiscent from Cook and Rackoff's Jumping Automata on Graphs~\cite{CR80}. 
A pointer graph is basically a multigraph structure representing the memory heap, whose nodes are references. The pointer graph semantics is designed to work on such a structure together with a stack, for method calls. The semantics is a high-level semantics whose purpose is to be independent from the bytecode or low-level semantics and will be defined on meta-instructions, a meta-language of flattened instructions with stack operations.

\subsection{Pointer graph}
A \emph{pointer graph} $\heap$ is a directed multigraph $(V,A)$. 
The nodes in $V$ are memory references and the arrows in $A$ link one reference to a reference of one of its fields.
Nodes are labeled by class names and arrows are labeled by the field name. In what follows, let $l_V$ be the node label mapping from $V$ to $\Class$ and $l_A$ be the arrow label mapping from $A$ to $\mathcal{F}$.

The memory heap used by a AOO program will be represented by a pointer graph. This pointer graph explicits the arborescent nature of objects: each constructor call will create a new node (memory reference) of the multigraph and arrows to the nodes (memory references) of its fields. Those arrows will be annotated by the field name. The heap in which the objects are stored corresponds to this multigraph. Consequently, bounding the heap memory use consists in bounding the size of the computed multigraph. 

\subsection{Pointer mapping}
A variable is of primitive (resp. reference) data type if it is declared using a type metavariable in $\mathbb{T}$ (resp. $\Class$).
A \emph{pointer mapping} with respect to a given pointer graph $\heap=(V,A)$ is a partial mapping $\pom: \Var \cup\{\this\} \mapsto V \cup \dom(\mathbb{T})$ associating primitive value in $\dom(\mathbb{T})$ to some variable of primitive data type in $\Var$ and a memory reference in $V$ to some variable of reference data type or to the current object $\this$.

 As usual, the domain of a pointer mapping $\pom$ is denoted $\dom(\pom)$. 
 
 By completion, for a given variable $\xa \notin \dom(\pom)$, let $\pom(\xa)= \nul$, if $\xa$ is of type $\C$, $\pom(\xa)=0$, if $\xa$ is of type $\ent$, $\pom(\xa)=\false$, if $\xa$ is of type $\bool$, ...
 The use of completion is just here to ensure that the presented semantics will not get stuck.
 
 \begin{example}\label{ex2}
Consider the class $\verb!BList!$ of Example~\ref{ex1}. Let $\verb!A!$ be a class having two $\verb!BList!$ fields $\xa_1$ and $\xa_2$ and $\verb!B!$ be a class extending $\verb!BList!$. Consider the initialization instruction $\tt Init$ defined by:
{\tt
\begin{lstlisting}
  Init ::=    BList b := new BList();
              BList c := new BList(true,b);
              B d := new B(c);
              A e := new A(c,c); 
\end{lstlisting}
}
Figure~\ref{graph:constructors} illustrates the pointer graph associated to this sequence of object creations. The figure contains both the pointer graph of labeled nodes and arrows together with the pointer mapping whose domain is represented by boxed variables and whose application is symbolized by snake arrows.

\begin{figure}
\centering
\begin{tikzpicture}[node distance=22mm,
varp/.style={->,decorate,decoration={snake,post length=1mm}}, 
var/.style={draw,node distance=1.1cm}]
\node (e1) {$\verb!A!$};
\node (c1) [right of=e1] {$\verb!BList!$};
\node (d1) [below of=c1] {$\verb!B!$};
\node (b1) [right of=c1] {$\verb!BList!$};
\node (a1) [right of=b1] {$\&\nul$};
\draw[->] (e1) to[bend left] node [above] {$\xa_1$} (c1) ;
\draw[->] (e1) to[bend right] node [below] {$\xa_2$} (c1);
\draw[->] (d1) to node [right] {${\tt queue}$} (c1);
\draw[->] (c1) to node [above] {${\tt queue}$} (b1);
\draw[->] (b1) to node  [above] {${\tt queue}$} (a1);

\node[var] (b) [above of=b1] {b};
\node[var] (c) [above of=c1] {c};
\node[var] (d) [left of=d1] {d};
\node[var] (e) [above of=e1] {e};

\draw[varp] (b) to (b1);
\draw[varp] (c) to (c1);
\draw[varp] (d) to (d1);
\draw[varp] (e) to (e1);
\end{tikzpicture}
\caption{Example of pointer graph and pointer mapping}\label{graph:constructors}
\end{figure}
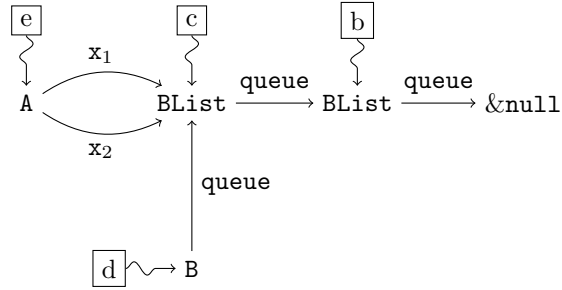
\end{example}
As we will see shortly, the semantics of an assignment $\xa:=\ea;$, for some variable $\xa$ of reference type, consists in updating the pointer mapping in such a way that $\pom(\xa)$ will be the reference of the object computed by $\ea$. By abuse of notation, let $\pom(\ea)$ be a notation for representing this latter reference.

In what follows, let $\pom[\xa \mapsto v]$, $v \in V \cup \dom(\mathbb{T})$, be a notation for the pointer mapping $\pom'$ that is equal to $\pom$ but on $\xa$ where the value is updated to $v$.

\subsection{Pointer stack}
For a given pointer graph $\heap$, a \emph{stack frame} $s_{\heap}$ is a pair $\langle s , \pom \rangle$ composed by a method signature $s=\tau\ \m^\C(\tau_1,\ldots,\tau_n)$ and a pointer mapping $\pom$.

A \emph{pointer stack} $\stack_\heap$ is a LIFO structure of stack frames corresponding to the same pointer graph $\heap$. Define $\top \stack_\heap$ to be the top pointer mapping of $\stack_\heap$. Finally, define $\pop(\stack_\heap)$ to be the pointer stack obtained from $ \stack_\heap$ by removing $\top \stack_\heap$ and $\push(s_{\heap},\stack_\heap)$ to be the pointer stack obtained by adding $s_{\heap}$ to the top of $\stack_\heap$.

In what follows, we will write just $\pop$ and $\push(s_{\heap})$ instead of $\pop(\stack_\heap)$ and $\push(s_{\heap},\stack_\heap)$ when $\stack_\heap$ is clear from the context.

As expected, the pointer stack of a program is used when calling a method: references to the parameters are pushed on a new stack frame at the top of the pointer stack. The pointer mappings of a pointer stack $\stack_\heap$ map method parameters to the references of the arguments on which they are applied, respecting the dynamic binding principle found in Object Oriented Languages.  
\begin{example}\label{ex3}
For example, considering the method $\verb!setQueue(BList q)!$ defined in the class $\verb!BList!$ of Example~\ref{ex1}, adding a method call $\verb!d.setQueue(b);!$ at the end of the initialization instruction of Example~\ref{ex2} will push a new stack frame $\langle \void\ {\verb!setQueue!}^{\verb!BList!} (\verb!BList!),\pom \rangle$ on the pointer stack. $\pom(\verb!q!)$ will point to $\pom(\verb!b!)$, the node corresponding to the object computed by $\verb!b!$, and $\pom(\this)$ will point to $\pom(\verb!d!)$, the node corresponding to the object computed by $\verb!d!$. 
\end{example}
At the beginning of an execution, the pointer stack will only contain the stack frame $\langle \void\ \main^{\Exe}(), p_0 \rangle$ ; $p_0$ being a mapping associating each local variable in the main method to the $\nul$ reference, whether it is of reference type, and to the basic primitive value otherwise ($\false$ for $\bool$ and $0$ for $\ent$). 
We will see shortly that a pop operation removing the top pointer mapping from the pointer stack will correspond, as expected, to the evaluation of a return statement in a method body.

\subsection{Memory configuration}

A memory configuration $\conf=\langle \heap, \stack_\heap \rangle$ consists in a pointer graph $\heap$ together with a pointer stack $\stack_\heap$. Among memory configurations, we distinguish the \emph{initial configuration} $\conf_{0}$ defined by $\conf_{0}=\langle (\{\& \nul\},\emptyset), [\langle \void\ \main^{\Exe}(), p_0 \rangle] \rangle$ where $\& \nul$ is the reference of the $\nul$ object (\textit{i.e.} $l_V(\& \nul)= \nul$) and $[s_\heap]$ denotes the pointer stack composed of only one stack frame $s_\heap$.

In other words, the initial configuration is such that the pointer graph only contains the null reference as node and no arrows and a pointer stack with one frame for the $\main$ method call.

\subsection{Meta-language and flattening}\label{flatflat}

The semantics of AOO programs will be defined on a meta-language of expressions and instructions. Meta-expressions are flat expressions. Meta-instructions consist in flattened instructions and \texttt{pop} and \texttt {push} operations for managing method calls. Meta-expressions and meta-instructions are defined formally by the following grammar:
$$\begin{array}{lllll}
\mea &\rgl & 
\xa \ |\ \n \ |\ \nul \ | \ \this   \ |\ \op(\overline{\xa}) \\
& &|\ \new\ \C(\overline{\xa}) \ |\ \xb.\m(\overline{\xa})
\end{array}$$
$$\begin{array}{lllll}
\mi &\rgl & 
 ; \ | \ [\tau]\ \xa \iasg \mea ;  \ | \   \mi_1 \ \mi_2 \ |\ \xb.\m(\overline{\xa});  
\\
 & & | \ \iwh(\xa) \ido \{ \mi  \} \ | \ \iif (\xa)\ithen\{ \mi_1\} \ielse \{\mi_2\} \\
 & & | \   \pop; \ | \ \push(s_\heap);\ | \ \epsilon
\end{array}$$
where $\epsilon$ denotes the empty meta-instruction.

Flattening an instruction $\ca$ into a meta-instruction $\underline{\ca}$ will consist in adding fresh intermediate variables for each complex expression parameter. This procedure is standard  and defined in Figure~\ref{flat}.
\begin{figure*}[!t]
\hrulefill
\centering 
\begin{align*}
\underline{[[\tau]\ \xa \iasg \ea];} &=[[\tau]\ \xa \iasg \ea]; \quad \text{if }\ea \in \Var \cup \dom(\mathbb{T}) \cup \{\this,\nul\}\\
\underline{[\tau]\ \xa \iasg \op(\eaa,\ldots,\ean);}&=\underline{\tau_1\ \xaa \iasg \eaa; \ldots \tau_n\ \xan \iasg \ean;}\ [\tau]\ \xa =\op(\xaa,\ldots,\xan);
\\
\underline{[\tau]\ \xa \iasg \new\ \C(\eaa,\ldots,\ean);}&=\underline{\tau_1\ \xaa \iasg \eaa; \ldots \tau_n\ \xan \iasg \ean;}\ [\tau]\ \xa \iasg \new \ \C(\xaa,\ldots,\xan); 
\\
[\underline{[\tau]\ \xa \iasg] \ea.\m(\eaa,\ldots,\ean);}&=\underline{\tau_{n+1}\ \xa_{n+1}=\ea; \tau_1\ \xaa \iasg \eaa; \ldots \tau_n\ \xan \iasg \ean;}\\
&\phantom{=} \ [[\tau]\ \xa \iasg]\xa_{n+1}.\m(\xaa,\ldots,\xan);
\\
\underline{ \ca_1\ \ca_2}&=\underline{\ca_1} \ \underline{\ca_2}\\
\underline{\iwh(\ea)\ido \{ \ca  \} } &= \underline{\bool \ \xaa := \ea;}\ \iwh(\xaa)\ido \{ \underline{\ca} \ \underline{\xaa := \ea;}\} 
\\
\underline{\iif (\ea)\ithen\{ \ca_1\} \ielse \{\ca_2\}} &= \underline{\bool \ \xaa := \ea;}\ \iif (\xaa)\ithen\{ \underline{\ca_1}\} \ielse \{\underline{\ca_2}\}
\end{align*}
{All $\xai$ represent fresh variables and the types $\tau_i$ match the expressions $\eai$ types}
\caption{Instruction flattening}
\label{flat}
\hrule
\end{figure*}
 The flattened meta-instruction will keep the semantics of the initial instruction unchanged. The main interest in such a program transformation is just that all the variables will be statically defined in a meta-instruction whereas they could be dynamically created by an instruction, hence allowing a cleaner (and easier) semantic treatment of meta-instructions. We extend the flattening to methods (and constructors) by $\tau \ \m(\tau_1\ \xa_1,\ldots,\tau_n\ \xa_n)\{ \underline{\ca} \ [\return \ \xa;]\}$ so that each instruction is flattened. A flattened program $\underline{P}$ is the program obtained by flattening all the instructions in the methods of a program $P$. Notice that the flattening of an AOO program is also an AOO program, as the flattening is a closed transformation with respect to the AOO syntax, and that the flattening is a polynomially bounded program transformation.
\begin{lemma}\label{lem:flat}
Define the size of an instruction $\size{\ca}$  (respectively meta-instruction $\size{\mi}$) to be the number of symbols in $\ca$ (resp. $\mi$). For each instruction $\ca$, we have $\size{\underline{\ca}} = O(\size{\ca}^2)$. 
\end{lemma}
\begin{proof}
By induction on the definition of flattening. The flattening of each atomic instruction adds a number of new symbols that is at most linear in the size of the original instruction. The number of instructions is also linear in $\size{\ca}$. 
\end{proof}
\begin{corollary}\label{lem:flatsize}
Define the size of an AOO program $\size{P}$ to be the number of symbols in $P$. For each program $P$, we have $\size{\underline{P}} = O(\size{P}^2)$. 
\end{corollary}
An alternative choice would have been to restrict program syntax by requiring expressions to be flattened, thus avoiding the use of the meta-language. However such a choice would impact negatively the expressivity of this study. On another hand, one possibility might have been to generate local variables dynamically in the programming semantics but such a treatment makes the analysis of pointer mapping domain (i.e. the number of living variables) a very hard task to handle.

\begin{example}\label{ex4}
We add the method $\verb!decrement()!$ to the class $\verb!BList!$ of Example~\ref{ex1}:
{\tt
\begin{lstlisting}
  void decrement() { 
    if (value) {
      value := false; 
    }
    else{
      if (queue != null) {
        value := true;
        queue.decrement();
      } else {
        value := false; 
      }
    }
  } 
\end{lstlisting}
}
The program flattening will generate the following body for the flattened method:
{\tt
\begin{lstlisting}
  void decrement() { 
    if (value) {
      value := false; 
    }
    else{
      BList $\xa_1$ := null;
      BList $\xa_2$ := queue; 
      boolean $\xa_3$ := $\xa_2$ != $\xa_1$
      if ($\xa_3$) {
        value := true;
        queue.decrement();
      } else {
        value := false; 
      }
    }
  } 
\end{lstlisting}
}
where $\xa_1$, $\xa_2$ and $\xa_3$ are fresh variables.
\end{example}

\subsection{Program semantics}\label{seman}

 \begin{figure*}[!t]
\hrulefill
 \begin{align}
&(\conf, [\tau]\ \xa \iasg \xb;\mi ) \to  ( \conf[\xa \mapsto \conf(\xb)], \mi ) \quad \textit{if } \xa \notin \Class(\conf).\mathcal{F} \\
&(\conf, [\tau]\ \xa \iasg \xb;\mi ) \to  ( \conf[\conf(\this) \ \stackrel{ \xa}{\mapsto} \conf(\xb)], \mi ) \quad \textit{if } \left\{\begin{array}{l}
\xa \in \Class(\conf).\mathcal{F}\\
\conf(\this) \neq \nul
\end{array}\right.\\
&(\conf, [\tau]\ \xa \iasg \xb;\mi ) \to  ( \conf, \mi ) \quad \textit{if } \xa \in \Class(\conf).\mathcal{F} \textit{ and } \conf(\this) = \nul\\
&(\conf, [\tau]\ \xa \iasg \n;\mi ) \to  ( \conf[\xa \mapsto \n], \mi )  \\
& (\conf, [\tau]\ \xa \iasg \nul;\mi ) \to  ( \conf[\xa \mapsto\ \&\nul], \mi ) \\
&(\conf, [\tau]\ \xa \iasg \this;\mi ) \to  ( \conf[\xa \mapsto \conf(\this)], \mi ) \\
&(\conf, [\tau]\ \xa \iasg \op(\overline{\xb});\mi ) \to  (\conf[\xa \mapsto \sem{\op}(\overline{\conf(\xb)})], \mi ) \\
&(\conf, [\tau]\ \xa \iasg \new\ \C(\overline{\xb});\mi ) \to  (\conf[v  \mapsto \C][v \ \stackrel{ \xc_i}{\mapsto} \conf(\xb_i) ][ \xa \mapsto v], \mi )  \\
&\quad\text{where } v \text{ is a fresh node (memory reference) and }\C.\mathcal{F}=\{\xc_1,\ldots,\xc_n\}  \notag\\
&(\conf, \ ; \mi) \to (\conf, \mi) \\
&(\conf, [\xd \iasg]\xa.\m(\overline{\xb});\mi ) \to  (\conf, \\
&\qquad\push(\langle \tau\ \m^{\C\star}(\overline{\tau}),\top \conf [\this \mapsto \conf(\xa), \xc_i \mapsto \conf(\xa.\xc_i), \xa_i \mapsto \conf(\xb_i)] \rangle); \notag\\
& \qquad\mi' \ [\xd \iasg \xa';]\ \pop;\ \mi ) \notag\\
& \quad\text{if }\m \text{ is a flattened method } \tau\ \m(\tau_1\ \xa_1,\ldots \tau_n\ \xa_n)\{ \mi' \ [\return \ \xa';]\} \text{ in }\C\star\nonumber \\
&\quad\text{and }\C.\mathcal{F}=\{\xc_1,\ldots,\xc_n\}  \notag\\
&(\conf, \push(s_\heap);\mi ) \to  (\conf[\push(s_\heap)], \mi ) \\
&(\conf, \pop;\mi ) \to  (\conf[ \pop], \mi ) \\
&(\conf,   \iwh(\xa)\ido \{ \mi'  \}\ \mi ) \to (\conf, \ \mi'\ \iwh(\xa)\ido \{ \mi'  \} \ \mi ) \quad \text{if }\conf(\xa)=\true\\
&(\conf,   \iwh(\xa)\ido \{ \mi'  \}\ \mi ) \to ({\conf, \ \mi })\quad  \text{if }\conf(\xa)=\false\\
&(\conf, \iif (\xa)\ithen\{ \mi_{\true}\} \ielse \{\mi_{\false}\}\ \mi)  \to (\conf, \mi_{\conf(\xa)}\ \mi)
\end{align}
\caption{Semantics of AOO programs\label{fig:sem}}
\hrulefill
\end{figure*}

Informally, the small step semantics  $\to$ of AOO programs relates a pair $( \conf, \mi )$ of memory configuration $\conf$ and meta-instruction $\mi$  to another pair $( \conf', \mi' )$ consisting of a new memory configuration $\conf'$ and of the next meta-instruction $\mi'$ to be executed. Let $\to^*$  (respectively $\to^+$) be its reflexive and transitive (respectively transitive) closure. In the special case where $ ( \conf, \mi ) \to^* ( \conf', \epsilon)$, we say that $\mi$ \emph{terminates on memory configuration} $\conf$.

Now, before defining the formal semantics of AOO programs, we introduce some preliminary notations. 

Given a memory configuration $\conf=\langle \heap, \push(\langle \tau\ \m^\C(\overline{\tau}), \pom\rangle ,\stack_\heap )\rangle$:
\begin{itemize}
\item $\top \conf = \pom$
\item $\conf(\xa)= \pom(\xa)$
\item $\conf(\xa.\xc)= v$, with $v \in \heap$ such that $\conf(\xa) \stackrel{ \xc}{\mapsto}v \in \heap$
\item $\Class(\conf)=\C$
\item $\conf[\xa \mapsto v]=\langle \heap, \push(\langle \tau \m^\C(\overline{\tau}), \pom[\xa \mapsto v]\rangle ,\stack_\heap )\rangle$, $v \in V \cup \dom(\mathbb{T})$
\item $\conf[ \push(s_\heap)]=\langle \heap, \push(s_\heap, \push(\langle \tau \m^\C(\overline{\tau}), \pom\rangle ,\stack_\heap ))\rangle$
\item $\conf[\pop]=\langle \heap, \stack_\heap \rangle$
\end{itemize}
In other words,  $\top \conf$ is a shorthand for the pointer mapping at the top of the stack $\push(\langle \tau \m^\C(\overline{\tau}), \pom\rangle ,\stack_\heap )$. $\conf(\xa)$ is a shorthand notation for $\top \conf(\xa)$. $\conf(\xa.\xc)$ is the memory reference of the field $\xc$ of the object stored in $\xa$. 
$\Class(\conf)$ is the class of the current object under evaluation in the memory configuration. $\conf[\xa \mapsto v]$, $v \in V \cup \dom(\mathbb{T})$, is a notation for the memory configuration $\conf'$ that is equal to $\conf$ but on $\top \conf (\xa)$ where the value is updated to $v$. $\conf[ \push(s_\heap)]$ and  $\conf[\pop]$ are notations for the memory configuration where the frame $s_\heap$ has been pushed to the top of the pointer stack and where the top pointer mapping has been removed from the top of the stack, respectively.

Finally, let $\conf[ v  {\mapsto} \C]$, $v \in V$ denote a memory configuration $\conf'$ whose graph contains the new node $v$ labeled by $\C$ (i.e. $l_V(v)=\C$) and let $\conf[ v \stackrel{\xa}{\mapsto} w]$, $v,w \in V$, denote a memory configuration $\conf'$ whose pointer graph contains the new arrow $(v,w)$ labeled by $\xa$ (\textit{i.e.} $l_A((v,w))=\xa$). In the case where there was already some arrow $(v,u)$ labeled by $\xa$ (\textit{i.e.} $l_A((v,u)) = \xa$) in the graph then it is deleted and replaced by $(v,w)$.

Each operator $\op \in \Operator$ of signature  $\op::\tau_1 \times \cdots \times \tau_n \to \tau_{n+1}$
is assumed to compute a total function $\sem{\op}: \tau_1 \times \cdots \times \tau_n \to \tau_{n+1}$ fixed by the language implementation.  Operators will be discussed in more details in Subsection~\ref{operators}.

The formal rules of the small step semantics are provided in Figure~\ref{fig:sem}. Let us explain the meaning of these rules.
\begin{itemize}
\item Rules (1) to (7) are fairly standard rules only updating the top pointer mapping $\top \conf$ of the considered configuration but do not alter the other parts of the pointer stack. Rule (2) is the only rule altering the pointer graph. Rule (1) describes the assignment of a variable to another (that is not a field of the current object). 
Rule (2) describes a field assignment when the object is not null. Notice that in this case, the pointer graph is updated as the reference node of the current object points to the reference node of the assigned variable. Rule (3) indicates that if the object is null, assigning its fields results in skipping the instruction. Rule (4) is the assignment of a primitive constant to a variable. Rule (5) is the assignment of the null reference $\&\nul$ to a variable.  
Rule (6) consists in the assignment of the self-reference. Notice that such an assignment may only occur in a method body (because of well-formedness assumptions) and consequently the stack is non-empty and must contain a reference to $\this$. Rule (7) describes operator evaluation based on the prior knowledge of function $\sem{\op}$. 
\item Rule (8) consists in the creation of a new instance. Consequently, this rule adds a new node $v$ of label $\C$ and the corresponding arrows $(v,\conf(\xb_i))$ of label $\xc_i$ in the pointer graph of $\conf$. $\conf(\xb_i)$ are the nodes of the graph (or the primitive values) corresponding to the parameters of the constructor call and $\xc_i$ is the corresponding field name in the class $\C$. Finally, this rule adds a link from the variable $\xa$ to the new reference $v$ in the pointer mapping $\top \conf$.
\item Rule (9) just consists in the evaluation of the empty instruction '$;$'. 
\item Rule (10) consists in a call to method $\m$ of class $\C$, provided that $\xc$ is of class $\C$, dynamically for overrides. If $\m$ is not defined in $\C$ then it checks for a method of the same signature in the upper class $\D$, and so on. Let $\C\star$ be a notation for the least superclass of $\C$ in which $\m$ is defined. It adds a new instruction for pushing a new stack frame on the stack, containing references of the current object $\this$ on which $\m$ is applied, references of each field, and references of the parameters. After adding the flattened body $\mi'$ of $\m$ to the evaluated instruction, it adds an assignment storing the returned value $\xc'$ in the assigned variable $\xa$, whenever the method is not a procedure, and a $\pop;$ instruction. The $\pop$ instruction is crucial to indicate the end of the method body so that the callee knows when to return the control flow to the caller.
\item Rules (11) and (12) are standard rules for manipulating the pointer stack through the use of $\pop$ and $\push$ instructions.
\item Rules (13) to (15) are standard rules for control flow statements. 
\end{itemize}

One important point to stress is that, contrarily to usual OO languages, the above semantics can reduce when a method is called on a variable pointing to the $\null$ value (see Rule (10)). In this particular case, all assignments to the fields of the current object are ignored (see Rule (3)). This choice has been made to simplify the formal treatment as exceptions are not included in AOO syntax. In general, exception handling generate a new control flow that is tedious to handle with a small step semantics. More, our typing discipline relies on the control flow being smooth.

\subsection{Input and Size}\label{input}
In order to analyze the complexity of programs, it is important to relate the execution of a given program to the size of its input. Consequently, we need to define both the notion of input of an AOO program and the notion of size for such an input. This will make clear the choice made in Section~\ref{syn} to express an executable class by splitting the instruction of the $\main$ method in an {\emph{initialization instruction}} $ \texttt{Init}$ and a {\emph{computational instruction}} $\texttt{Comp}$.

An AOO program of executable $\Exe\{ \void \ \main()\{ \texttt{Init}\ \texttt{Comp} \} \}$ terminates if the following conditions hold:
\begin{enumerate}
\item $(\conf_0,\underline{\texttt{Init}} )\to^* (\data,\epsilon)$
\item $(\data,\underline{\texttt{Comp}} )\to^* (\conf,\epsilon) $
\end{enumerate}
The memory configuration $\data$ computed by the initialization instruction is called the \emph{input}.

An important point to stress is that, given a program, the choice of initialization and computational instructions is left to the analyzer. This choice is crucial for this analysis to be relevant. Note also that the initialization instruction {\tt Init} can be changed so that the program may be defined on any input. 

Another way to handle the notion of input could be to allow i/o (open a stream on a file, deserialization,...) in this instruction but at the price of a more technical and not that interesting treatment. A last possibility is to consider the input to be the main method arguments. But this would mean that as most of the programs do not use their prompt line argument, they would be considered to be constant time programs. 

There are two particular cases:
\begin{itemize}
\item If the initialization instruction is empty then there will be no computation on reference type variables apart from constant time or non-terminating ones, as we will see shortly. This behavior is highly expected as it means that the program has no input. As there is no input, it means that either the program does not terminate or it terminates in constant time.
\item If the computational instruction is empty (that is ``;'') then the program will trivially pass the complexity analysis.
\end{itemize}

\begin{example}\label{ex5}
Consider the initialization instruction ${\tt Init}$ of Example~\ref{ex2}. ${\tt Init}$ is already flattened so that $\underline{{\tt Init}}={\tt Init}$ holds. The input $\mathcal{I}$ is a memory configuration $\langle \heap , \stack_\heap \rangle$ such that $\heap$ is the pointer graph described in Figure~\ref{graph:constructors} and $\stack_\heap = [\langle \void\ \main^{\Exe}(), p_\heap \rangle]$, $p_\heap$ being the pointer mapping described by the snake arrows of Figure~\ref{graph:constructors}.

Now consider the computational instruction ${\tt Comp}=\verb!d.setQueue(b);!$. It is also flat and so is the method body. Consequently, we have:
\begin{align*}
(\mathcal{I},{\tt d.setQueue(b);}) &\to (\mathcal{I}, \push(s_\heap); {\tt queue} \iasg {\tt q};\ \pop;)\\
&\to (\mathcal{I}[\push(s_\heap)] , {\tt queue} \iasg {\tt q};\ \pop;)\\
&\to (\mathcal{I'}[\mathcal{I'}(\this) \stackrel{{\tt queue}}{\mapsto} \mathcal{I'}({\tt q})], \pop;)\\
&= (\mathcal{I'}[\mathcal{I}({\tt d}) \stackrel{{\tt queue}}{\mapsto} \mathcal{I}({\tt b})], \pop;)\\
&\to (\mathcal{I'}[\mathcal{I}({\tt d}) \stackrel{{\tt queue}}{\mapsto} \mathcal{I}({\tt b})][\pop], \epsilon)\\
&=(\mathcal{I}[\mathcal{I}({\tt d}) \stackrel{{\tt queue}}{\mapsto} \mathcal{I}({\tt b})], \epsilon)
\end{align*}
with $s_\heap =\langle \void\  \verb!setQueue!^{\verb!BList!}(\verb!BList!),\top \mathcal{I} [\this \mapsto \mathcal{I}(\verb!d!), \verb!q! \mapsto \mathcal{I}(\verb!b!)] \rangle$ and $\mathcal{I'}=\mathcal{I}[\push(s_\heap)]$. The two first transitions are an application of Rules (8) and (9) of Figure~\ref{fig:sem}. The third one is Rule (14) as ${\tt queue}$ is a field of the current object. The first equality holds by definition of $\mathcal{I'}$ and its top stack frame $\top \mathcal{I'}=s_\heap$. The last reduction is Rule (10) and the equality holds as removing the top stack of $\mathcal{I'}$ leads to $\mathcal{I}$, keeping in mind that the pointer graph has been updated by $[\mathcal{I}({\tt d}) \stackrel{{\tt queue}}{\mapsto} \mathcal{I}({\tt b})]$. Indeed now the node referenced by $\tt d$ points to $\tt b$ via the arrow labeled by $\tt queue$.
\end{example}

\begin{definition}[Sizes] The size:
\begin{enumerate}
\item $\size{\heap}$ of a pointer graph $\heap=(V,A)$ is defined to be the number of nodes in $V$.\label{i1}
\item $\size{\pom}$ of a pointer mapping $\pom$  is defined by $\size{\pom}= \sum_{\xa \in \dom(\pom)}\size{\pom(\xa)}$, where the size of numerical primitive value is the value itself, the size of a boolean is $1$ and the size of a memory reference is $1$. \label{i2}
\item $\size{s_\heap}$ of a stack frame $s_\heap=\langle s, p \rangle$ is defined by: $\size{s_\heap}=1+\size{p}$. 
\item $\size{\stack_\heap}$ of a stack $\stack_\heap$ is defined by $\size{\stack_{\heap}}=\sum_{s_{\heap} \in \stack_{\heap}}\size{s_{\heap}}$.\label{i4}
\item $\size{\conf}$ of a memory configuration $\conf=\langle \heap,\stack_{\heap}\rangle$ is defined by $\size{\conf}=\size{\heap}+\size{\stack_{\heap}}$.\label{i5}
\end{enumerate}
\end{definition}

In Item~\ref{i1}, it suffices to bound the number of nodes as, for the considered multigraphs, $\taille{A}=\mathcal{O}(\taille{V})$. Indeed, each node has a number of out arrows bounded by a constant (the maximal number of fields in all classes of a given program). In Item~\ref{i2}, numerical primitive values are not considered to be constant. This definition is robust if we consider their size to be constant: \text{e.g.} a signed 32 bit integer could be considered as a constant smaller than $2^{31}-1$. In such a case, the size of each pointer mapping would be constant as no fresh variable can be generated within a program thanks to flattening.
In Item~\ref{i4}, the size of a pointer stack is very close to the size of the usual OO Virtual Machine stack since it counts the number of nested method calls (i.e. the number of stack frames in the stack) and the size of primitive data in each frame (that are duplicated during the pass-by-value evaluation). 
Finally, Item~\ref{i5} defines the size of a memory configuration, the program input and bounds both the heap size of the input $\data$ and \textit{a fortiori} the stack size as the stack is empty in $\data$.

Negative integers are not considered for simplicity reasons (see section~\ref{operators} where we need data types to have a lower bound). Floats are not considered as there is an infinite number of floats of the same size. Only finite types and countable types can be treated. To our knowledge this point is not tackled by the related works on {\sc icc}.

\subsection{Compatible pairs}\label{compatible}
Given a memory configuration $\conf$ and a meta-instruction $\mi$, the pair $(\conf,\mi)$ is compatible if there exists an instruction $\mi'$ from a well-formed program such that $(\conf_0,\mi') \to^* (\conf, \mi)$. Throughout the paper, we will only consider compatible pairs $(\conf,\mi)$. 

This notion is introduced in order to prevent the consideration of a pair $(\conf,\mi)$ having a variable not defined in the memory configuration $\conf$ and called without being declared in the meta-instruction $\mi$.

Note that the semantics of Section~\ref{seman} cannot reach uncompatible pairs as each variable is supposed to be declared before its first use by definition of well-formed programs. However, this restriction will be required in order to prevent bad configurations from being considered in Theorem~\ref{thm:soundness}.

\section{Type system}\label{type}
In this section, a tier based type system for ensuring polynomial time and polynomial space upper bounds on the size of a memory configuration is introduced. 

\subsection{Tiered types}
The set of base types is defined to be the set of all primitive and reference types $\Class \cup \mathbb{T}$. \emph{Tiers} are elements of the lattice $(\{\tiera,\tierb\},\join,\wedge)$ where $\wedge$ and $\join$ are the greatest lower bound operator and the least upper bound operator, respectively ; the induced order, denoted $\ord$, being such that $\tiera \ord \tierb$. Given a sequence of tiers $\overline{\alpha}=\alpha_1,\ldots, \alpha_n$ define $\vee \overline{\alpha}=\alpha_1 \vee \ldots \vee \alpha_n$. Let $\alpha,\beta,\gamma,\ldots$ denote tiers in $\{\tiera, \tierb\}$.

A \emph{tiered type} is a pair $\tau(\alpha)$ consisting of a type $\tau \in \Class \cup \mathbb{T}$ together with a tier $\alpha \in \{\tiera, \tierb\}$.

\paragraph{Notations} Given two sequences of types $\overline{\tau}=\tau_1,\ldots,\tau_n$ and tiers $\overline{\alpha}=\alpha_1,\ldots, \alpha_n$ and a tier $\alpha$, let $\overline{\tau}(\overline{\alpha})$ denote $\tau_1(\alpha_1),\ldots,\tau_n(\alpha_n)$,  $\overline{\tau}({\alpha})$ denote $\tau_1(\alpha),\ldots,\tau_n(\alpha)$ and  $\langle \overline{\tau}\rangle$ (resp. $\langle \overline{\tau}(\overline{\alpha})\rangle$) denote the cartesian product of types (resp. tiered types). 

For example, given a sequence of types $\overline{\tau}=\bool, \charac, \C$ and a sequence of tiers $\overline{\alpha}=\tiera,\tierb,\tierb$, we have $\overline{\tau}(\tiera)=\bool(\tiera), \charac(\tiera), \C(\tiera)$, $\overline{\tau}=\bool \times \charac \times \C$ and $\langle \overline{\tau}(\overline{\alpha})\rangle= \bool(\tiera) \times \charac(\tierb) \times \C(\tierb)$.

\paragraph{What is a tier in essence} Tiers will be used to separate data in two kinds as in Bellantoni and Cook's safe recursion scheme~\cite{BC92} where data are divided into ``safe'' and ``normal'' data kinds. According to Danner and Royer~\cite{DN12},  ``normal data [are the data] that drive recursions and safe data [are the data] over which recursions compute''. In this setting, tier $\tierb$ will be an equivalent for normal data type, as it consists in data that drive recursion and while loops. Tier $\tiera$ will be the equivalent for safe data type, as it consists in computational data storages. Instruction tiers will ensure that expressions of the right tier are used at the right place (e.g. tier $\tiera$ data will never drive a loop or a recursion) and that the information flow never go from $\tierb$ to $\tiera$ so that the first condition is preserved independently of tier $\tiera$ during program execution. In particular, a tier $\tierb$ instruction will never be controlled by a tier $\tiera$ instruction (in a conditional statement or recursive call).

 \subsection{Operators}\label{operators}
 In order to control the complexity of programs, we need to constrain the admissible tiered types of operators  depending on their computational power. For that purpose and following~\cite{HMP13}, we define \emph{neutral operators} whose computation does not make the size increase and \emph{positive operators} whose computation may make the size increase by some constant. They are both assumed to be polynomial time computable.
 \begin{definition}An operator  $\op::\tau_1 \times \ldots \times \tau_n \to \tau$ is:
\begin{enumerate}
\item \emph{neutral} if $\sem{\op}$ is a polynomial time computable function and one of the following conditions hold:
\begin{enumerate}
\item either $\tau=\bool$ or $\tau =\charac$,\label{a}
\item $\op::\ent \times \ldots \times \ent \to \ent$ and $\tt \forall i \in [1,n],\ \forall \textnormal{\tt cst}_{\ent}^i:$
$$ 0 \leq \sem{\op}(\textnormal{\tt cst}_{\ent}^1,\ldots, \textnormal{\tt cst}_{\ent}^n) \leq \max_{j}{\textnormal{\tt cst}_{\ent}^j}.$$ \label{b}
\end{enumerate} 
\item \emph{positive} if it is not neutral, $\sem{\op}$ is a polynomial time computable function, $\op::\ent \times \ldots \times \ent \to \ent$, and:
$$\tt \forall i \in [1,n],\ \forall \textnormal{\tt cst}_{\ent}^i,\ 0 \leq \sem{\op}(\textnormal{\tt cst}_{\ent}^1,\ldots, \textnormal{\tt cst}_{\ent}^n) \leq \max_{j}{\textnormal{\tt cst}_{\ent}^j}+c,$$ for some constant ${\tt c} \geq 0$.\label{c}
\end{enumerate}
\end{definition}
In the above definition, Item~\ref{a} could be extended to any primitive data type inhabited by a finite number of values. 

In Items~\ref{b} and~\ref{c}, the comparison is only defined whenever all the $\tau_i$ and $\tau$ are of type $\ent$. This could be extended to booleans without any trouble by considering a boolean to be $0$ or $1$. 

Also notice that operator overload can be considered by just treating two such operators as if they were distinct ones.

\begin{example}\label{ex6}
For simplicity, operators are always used in a prefix notation.
The operator $-$, whose semantics is such that $\sem{-}(x,y)=\max(x-y,0)$, is neutral. Indeed for all $\textnormal{\tt cst}_{\ent}^1$, $\textnormal{\tt cst}_{\ent}^2$, $\sem{-}(\textnormal{\tt cst}_{\ent}^1,\textnormal{\tt cst}_{\ent}^2)=\max(\textnormal{\tt cst}_{\ent}^1-\textnormal{\tt cst}_{\ent}^2,0) \leq \max(\textnormal{\tt cst}_{\ent}^1,\textnormal{\tt cst}_{\ent}^2)$. The operators $<$, testing that its left operand is strictly smaller than its right one, and the operator $\verb!==!$, testing the equality of primitive values or the equality of memory references, and the operator $\tt !=$ of Example~\ref{ex4} are neutral operators as their output is $\bool$ and they are computable in polynomial time.

The operator $+$ is neither neutral, nor positive as $\sem{+}(\textnormal{\tt cst}_{\ent}^1,\textnormal{\tt cst}_{\ent}^2)=\textnormal{\tt cst}_{\ent}^1+\textnormal{\tt cst}_{\ent}^2$ and there is no constant $\textnormal{\tt c}\geq 0$ such that $\forall \textnormal{\tt cst}_{\ent}^1,\textnormal{\tt cst}_{\ent}^2,\ \textnormal{\tt cst}_{\ent}^1+\textnormal{\tt cst}_{\ent}^2 \leq \textnormal{\tt cst}_{\tau_i}^i+\textnormal{\tt c}$. However if we consider its partial evaluation $+n$ to be an operator, then $+n$ is a positive operator as $\forall \textnormal{\tt cst}_{\ent}^1,\ \sem{+n}(\textnormal{\tt cst}_{\ent}^1) = \textnormal{\tt cst}_{\ent}^1+n \leq \textnormal{\tt cst}_{\tau_i}^i+\textnormal{\tt c}$. Indeed, just take $\textnormal{\tt c} \geq n$. The reason behind such a strange distinction is that a call of the shape $\xa \iasg \xb + \xb$ could lead to an exponential computation in a loop while a call of the shape $\xa \iasg \xb + n$ should remain polynomial (under some restrictions on the loop).

As mentioned above, the notions of neutral and positive operators can be extended to operators with numerical output and boolean input. Hence allowing to consider operators such as {\tt x?y:z} returning $\xb$ or $\xc$ depending on whether $\xa$ is $\true$ or $\false$ when applied to numerical primitive data. 
Indeed, we have $\llbracket${\tt ?:}$\rrbracket(\xa,\xb,\xc) \leq \max(\xb,\xc)$. Consequently, {\tt ?:} is a neutral operator.
\end{example}

\begin{definition}\label{typing:op}
An \emph{operator typing environments} $\Omega$ is a mapping such that each operator  $\op::\langle \overline{\tau} \rangle \to \tau$ is assigned the type:
\begin{itemize}
\item $\Omega(\op)=\{ \langle \overline{\tau}(\tiera) \rangle \to \tau(\tiera), \langle \overline{\tau}(\tierb) \rangle\to \tau(\tierb)\}$ if $\op$ is a neutral operator,
\item $\Omega(\op)=\{ \langle \overline{\tau}(\tiera) \rangle \to \tau(\tiera)\}$ if $\op$ is a positive operator,
\item $\Omega(\op)=\emptyset$ otherwise.
\end{itemize}
\end{definition}

\subsection{Environments}
In this subsection, we define three other kinds of environments:
\begin{itemize}
\item \emph{method typing environments} $\delta$ associating a tiered type to each program variable $v \in \Var$,
\item \emph{typing environments} $\Delta$ associating a typing environment $\delta$ to each method signature $\tau\ \m^\C(\overline{\tau})$, \text{i.e.} $\Delta(\tau\ \m^\C(\overline{\tau}))=\delta$ (by abuse of notation, we will sometime use the notation $\Delta(\m^\C)$ when the method signature is clear from the context),
\item \emph{contextual typing environments} $\typenv = (s, \Delta)$, a pair consisting of a method signature and a typing environment. The method (or constructor) signature $s$ indicates under which context (typing environment) the fields are typed.
\end{itemize}

For each $\xa \in \Var$, define $\typenv(\xa)=\Delta(s)(\xa)$. Also define $\typenv\{\xa \leftarrow \tau(\alpha)\}$ to be the contextual typing environment $\typenv'$ such that $\forall \xb \neq \xa, \ \typenv'(\xb)=\typenv(\xb)$ and $\typenv'(\xa)=\tau(\alpha)$. Let $\typenv\{s'\}$ be a notation for the contextual typing environment that is equal to $(s',\Delta)$, \textit{i.e.} the signature $s'$ has been substituted to $s$ in $\typenv$. Method typing environment are defined for method signatures and can be extended without any difficulty to constructor signatures.

\noindent\textit{What is the purpose of contextual environments.} They will allow the type system to type a field with distinct tiered types depending on the considered method. Indeed, a field can be used in the guard of a while loop (and will be of tier $\tierb$) in some method whereas it can store freshly created objects (and will be of tier $\tiera$) in some other method. This is the reason why the presented type system has to keep information on the context.

\subsection{Judgments}
Expressions and instructions will be typed using tiered types. A method of arity $n$ method will be given a type of the shape $\C(\beta)\times\tau_1(\sla_1) \times \ldots \times  \tau_n(\sla_n) \to \tau(\sla)$.
Consequently, given a contextual typing environment $\typenv$, there are four kinds of typing judgments:
\begin{itemize}
\item $\typenv,\Omega \imp_\beta \ea:\tau(\alpha)$ for expressions, meaning that the expression $\ea$ is of tiered type $\tau(\alpha)$ under the contextual typing environment $\typenv$ and operator typing environment $\Omega$ and can only be assigned to in an instruction of tier at least $\beta$,
\item $\typenv,\Omega \imp \ca:\void(\alpha)$ for instructions, meaning that the instruction $\ca$ is of tiered type $\void(\alpha)$ under the contextual typing environment $\typenv$ and operator typing environment $\Omega$,
\item $\typenv,\Omega \imp_\beta s :\C(\beta)\times \langle \overline{\tau}(\overline{\tya})\rangle \to \tau(\alpha)$ for method signatures, meaning that the method $\m$ of signature $s$ belongs to the class $\C$ ($\C(\beta)$ is the tiered type of the current object $\this$), has parameters of type $\langle \overline{\tau}(\overline{\tya})\rangle$, has a return variable of type $\tau(\alpha)$, with $\tau=\void$ in the particular case where there is no return statement, and can only be called in instructions of tier at least $\beta$,
\item $\typenv , \Omega \imp \C(\overline{\tau}) : \langle \overline{\tau}(\tiera)\rangle \to  \C(\tiera)$ for constructor signatures, meaning that the constructor $\C$ has parameters of type $\langle \overline{\tau}(\tiera)\rangle$ and a return variable of type $\C(\tiera)$, matching the class type $\C$.

\end{itemize}
 Given a sequence $\overline{\ea}=\ea_1,\ldots,\ea_n$ of expressions, a sequence of types $\overline{\tau}=\tau_1,\ldots,\tau_n$ and two sequences of tiers $\overline{\alpha}=\alpha_1,\ldots,\alpha_n$ and $\overline{\beta}=\beta_1,\ldots,\beta_n$, the notation $\typenv, \Omega \imp_{\overline{\beta}} \overline{\ea} : \overline{\tau}(\overline{\alpha})$ means that $\typenv,\Omega \imp_{\beta_i} {\ea}_i : {\tau}_i(\alpha_i)$ holds, for all $i \in [1,n]$.

\subsection{Typing rules}
The intuition of the typing discipline is as follows: keeping in mind, that tier $\tierb$ corresponds to while loop guards data and that tier $\tiera$ corresponds to data storages (thus possibly increasing data), the type system precludes flows from tier $\tiera$ data to tier $\tierb$ data.

Most of the rules are basic non-interference typing rules following Volpano \textit{et al.} type discipline: tiers in the rule premises (when there is one) are equal to the tier in the rule conclusion so that there can be no information flow (in both directions) using these rules.

 \subsubsection{Expressions typing rules}
The typing rules for expressions are provided in Figure~\ref{fig:TypeE}. They can be explained briefly:
\begin{itemize}
\item Primitive constants and the null reference can be given any tier $\alpha$ as they have no computational power (Rules \textit{(Cst)}  and \textit{(Null)}) and can be used in instructions of any tier $\beta$. As in Java and for polymorphic reasons, $\nul$ can be considered of any class $\C$.
\item Rule  \textit{(Var)} just consists in checking a variable tiered type with respect to a given contextual typing environment. This is the rule allowing a polymorphic treatment of fields depending on the context. Indeed, remember that $\typenv(\xa)$ is a shorthand notation for $\Delta(s)(\xa)$. Moreover, a variable can be used in instructions of any tier $\beta$. 
\item Rule  \textit{(Op)} just consists in checking that the type of an operator with respect to a given operator typing environment matches the input and output tiered types. The operator output has to be used in an instruction of tier being at least the maximal admissible tier of its operand $\vee \overline{\beta} = \beta_1 \vee \ldots \vee \beta_n$. In other words, if one operand can only be used in a tier $\tierb$ instruction then the operator output can only be used in a tier $\tierb$ instruction.
\item The rule \textit{(Self)} makes explicit that the self reference $\this$ is of type $\C$ and enforces the tier of the fields to be equal to the tier of the current object, thus preventing ``flows by references'' in the heap. Moreover it can be used in instructions of any tier $\beta$.
\item The rule \textit{(Pol)} allows us to type a given expression with the superclass type while keeping the tiers unchanged. 
\item The rule \textit{(New)}  checks that the constructor arguments and output all have tier $\tiera$. The new instance has to be of tier $\tiera$ since its creation makes the memory grow (a new reference node is added in the heap). The constructor arguments have also to be of tier $\tiera$. Otherwise a flow from tier $\tiera$, the new instance, to tier $\tierb$, one of its fields, might occur. The explanation for the annotation instruction tier $\beta$ is the same as the one for operators.
\item Rule \textit{(C)} just  checks a direct type correspondence between the arguments types and the method type when a method is called. However this rule is very important as it allows a polymorphic type discipline for fields. Indeed the contextual environment is updated so that a field can be typed with respect to the considered method.  The explanation for the annotation instruction tier $\beta$ is the same as the one for operators. We will see shortly how to make restriction on such annotations in order to restrict the complexity of recursive method calls.
\end{itemize}

 \begin{figure}[!ht]
{\small
$$
\ninfer{ }{\typenv, \Omega \imp_\beta \n:\tau(\alpha)}{\textit{(Cst)}}
\qquad
\ninfer{ }{\typenv, \Omega \imp_\beta \nul:\C(\alpha)}{\textit{(Null)}}
$$
$$
\ninfer{ \quad \typenv(\xa)=\tau(\tya) }{\typenv, \Omega\imp_\beta \xa:\tau(\tya)}{\textit{(Var)}}
$$
$$
\ninfer{\typenv, \Omega \imp_{\overline{\beta}} \overline{\ea}:\overline{\tau}(\alpha) \qquad \langle \overline{\tau}(\alpha) \rangle \to \tau (\alpha) \in \Omega(\op)}
{\typenv, \Omega \imp_{\vee \overline{\beta}} \op(\overline{\ea}):\tau(\alpha) }
{\textit{(Op)}}
$$
$$
\ninfer{\forall \xa \in \C.\mathcal{F}, \exists \tau, \typenv(\xa)=\tau(\alpha)
}{\typenv, \Omega \imp_\beta \this:\C(\alpha)}{\textit{(Self)}}
\qquad 
\ninfer{\typenv, \Omega \imp_\beta \ea : \D(\alpha) \quad \D \dord \C}{\typenv,\Omega \imp_\beta \ea : \C(\alpha)}
{\textit{(Pol)}}
$$
$$
\ninfer{\typenv, \Omega \imp_{\overline{\beta}} \overline{\ea}:\overline{\tau}(\tiera)  \quad \typenv\{\C(\overline{\tau})\}, \Omega \imp \C(\overline{\tau}): \langle \overline{\tau}(\tiera) \rangle \to \C(\tiera) } 
{\typenv, \Omega \imp_{\vee \overline{\beta}} \new\ \C(\overline{\ea}):\C(\tiera)}
{\textit{(New)}}
$$
$$
\ninfer{
\typenv,\Omega\imp_\beta \ea:\C(\tyb) \  \typenv,\Omega\imp_{\overline{\beta}} \overline{\ea}:\overline{\tau}(\overline{\tya})  \ 
  \typenv\{\tau\ \m^{\C}(\overline{\tau})\},\Omega\imp_\beta  \tau\ \m^{\C}(\overline{\tau}) : \C(\beta)\times \langle \overline{\tau}(\overline{\tya}) \rangle\to  \tau({\tya})
 } 
{\typenv,\Omega\imp_{\beta \vee (\vee \overline{\beta})}  \ea.\m(\overline{\ea}):\tau(\tya)}
{\textit{(C)}}
$$
\caption{Typing rules for expressions\label{fig:TypeE}}
}
\end{figure}

 \subsubsection{Instructions typing rules}
The typing rules for instructions are provided in Figure~\ref{fig:TypeI}. They can be explained briefly:
\begin{itemize}
\item Rule \textit{(Skip)} is straightforward.
\item Rule \textit{(Seq)} shows that the tier of the sequence $\ca_1\ \ca_2$ will be the maximum of the tiers of $\ca_1$ and $\ca_2$. It can be read as ``a sequence of instructions including at least one instruction that cannot be controlled by tier $\tiera$ cannot be controlled by tier $\tiera$'' and it preserves non-interference as it is a weakly monotonic typing rule with respect to tiers. 
\item The recovery is performed thanks to the Rule \textit{(ISub)} that makes possible to type an instruction of tier $\tiera$ by $\tierb$ (as tier $\tiera$ instruction use is less constrained than tier $\tierb$ instruction use) without breaking the system non-interference properties. Notice also that there is no counterpart for expressions as a subtyping rule from $\tierb$ to $\tiera$ would allow us to type $\verb!x+1;!$ with $\verb!x!$ of tier $\tierb$ while a subtyping rule from $\tiera$ to $\tierb$ would allow the programmer to type programs with tier $\tiera$ variables in the guards of while loops.
\item Rule \textit{(Ass)} is an important non-interference rule. It forbids information flows of reference type from a tier  to another: it is only possible to assign an expression $\ea$ of tier $\alpha$ to a variable $\xa$ of tier $\alpha$. In the particular case of primitive data, a flow from $\tierb$ to $\tiera$ might occur. This is just because primitive data are supposed to be passed-by-value (otherwise we should keep the equality of tier preserved).  In the case of a field assignment, all tiers are constrained to be $\tiera$ in order to avoid changes inside the tier $\tierb$ graph. Finally, if the expression $\ea$ cannot be used in instructions of tier less than $\beta$, we constrain the instruction tiered type to be $\void(\alpha \vee \beta)$ in order to fulfill this requirement.
\item Rule \textit{(If)} constrains the tier of the conditional guard $\ea$ to match the tiers of the branching instructions $\ca_1$ and $\ca_2$. Hence, it prevents assignments of tier $\tierb$ variables to be controlled by a tier $\tiera$ expression.
\item Rule \textit{(Wh)} constrains the guard of
the loop $\ea$ to be a boolean expression of tier $\tierb$, thus preventing while loops from being controlled by tier $\tiera$ expressions. 
\end{itemize}

\begin{figure}[!ht]
{
$$
\ninfer{ \phantom{truc}}{\typenv, \Omega\imp \ ;\ :\void(\alpha)}{\textit{(Skip)}}
\qquad 
\ninfer{\forall i,\ \typenv, \Omega\imp \ca_i:\void(\tya_i) } 
{\typenv, \Omega\imp \ca_1 \ \ca_2: \void(\tya_1 \join \tya_2)}
{\textit{(Seq)}}
$$
$$
\ninfer{\typenv, \Omega\imp \ca:\void(\tiera)} 
{\typenv, \Omega\imp \ca :\void(\tierb)}
{\textit{(ISub)}}
\qquad
\ninfer{[\typenv, \Omega\imp_{\tiera} \xa: \tau(\tya')]  \quad \typenv, \Omega\imp_\beta \ea :\tau(\tya)} 
{\typenv, \Omega\imp [[\tau]\ \xa \iasg] \ea; : \void(\tya \vee \beta)}
{\textit{(Ass**)}}
$$
$$
\ninfer{\typenv, \Omega\imp_\tya \ea :\bool(\tya) \quad \forall i,\ \typenv, \Omega\imp \ca_i: \void(\tya) } 
{\typenv, \Omega\imp \iif (\ea) \ithen \{\ca_1\} \ielse \{\ca_2\}:\void(\tya)}
{\textit{(If)}} 
$$
$$
\ninfer{\typenv, \Omega \imp_\tierb \ea:\bool(\tierb) \quad \typenv, \Omega\imp \ca:\void(\tierb)} 
{\typenv, \Omega\imp \iwhile (\ea) \ido \{\ca\} :\void(\tierb)}
{\textit{(Wh)}}  
$$
 $ \textit{(**)}$
 
 $\xa \in \mathcal{F} \implies \alpha=\tiera$
 
 $\tau \in \mathbb{T} \implies \alpha' \leq \alpha$
 
 $\tau \in \Class \implies \alpha' =\alpha$
\caption{Typing rules for instructions\label{fig:TypeI}}
}
\end{figure}

 \subsubsection{Methods typing rules}
The typing rules for methods are provided in Figure~\ref{fig:TypeCM}. They can be explained briefly:
\begin{itemize}
\item Rule \textit{(Body)} shows how to type method definitions.
It updates the environment with respect to the parameters, current object and return type in order to allow a polymorphic typing discipline for methods: while typing a program, a method can be given distinct types depending on where and how it is called.
The tier $\gamma$ of the method body is used as annotation so that if $\gamma=\tierb$, the method cannot be called in a tier $\tiera$ instruction. We also check that the current object matches the tier $\beta$. 
As presented, this rule may create an infinite branch in the typing tree of a recursive method.
This could be solved by keeping the set of methods previously typed in a third context, we chose not to include this complication. Consequently, a method is implicitly typed only once in a given branch of a typing derivation. 
\item Rule \textit{(Cons)} shows how to type constructors. The constructor body and parameters are enforced to be of tier $\tiera$ as a constructor invocation makes the memory grow.
\item Rule \textit{(OR)} deals with overridden method, keeping tiers preserved, thus allowing standard OO polymorphism.
\end{itemize}
\begin{figure}[!ht]
{\small
$$
\ninfer{ \typenv\{\this\leftarrow \C(\beta), \overline{\xa} \leftarrow \overline{\tau}(\overline{\tya}),  [\xa \leftarrow \tau(\alpha)] \}, \Omega \imp I : \void(\gamma) \quad \typenv,\Omega \imp_\gamma \this : \C(\beta)}
{\typenv, \Omega \imp_{\gamma} \tau\ \m^{\C}(\overline{\tau}) : \C(\beta) \times \langle \overline{\tau}(\overline{\tya})\rangle \to  \tau({\tya})}
{\textit{(Body**)}} 
$$
$$
\ninfer{ \typenv\{\overline{\xa} \leftarrow \overline{\tau}(\tiera) \} , \Omega \imp I : \void(\tiera) \quad \C(\overline{\tau\ \xa}) \{\ca \} \in \C }
{\typenv , \Omega \imp \C(\overline{\tau}) : \langle \overline{\tau}(\tiera)\rangle \to  \C(\tiera)}
{\textit{(Cons)}} 
$$
$$
\ninfer{\C \dord \D \qquad \typenv,\Omega \imp_\gamma \tau\ \m^{\D}(\overline{\tau}) : \D({\tyb})\times \langle \overline{\tau}(\overline{\tya}) \rangle\to  \tau({\tya}) \quad  \tau\ \m(\overline{\tau\ \xa}) \{\ca\
[\return\ {\xa};] \} \in \D}{\typenv,\Omega\imp_\gamma \tau\ \m^{\C}(\overline{\tau}) : \C({\tyb})\times\langle \overline{\tau}(\overline{\tya}) \rangle \to  \tau({\tya})}
{\textit{(OR)}}
$$
(**) provided that $\tau\ \m(\overline{\tau\ \xa}) \{\ca\
[\return\ {\xa};] \} \in \C $
}
\caption{Typing rules for methods}\label{fig:TypeCM}
\end{figure}
In the above rules, the tier $\gamma$ is just an annotation on the least tier of the instructions where the method is called. It will be used in the next Section to constrain the tier of recursive methods.

\subsection{Well-typedness.} 
Given a program of executable $\Exe\{ \void \  \main()\{ \texttt{Init}\ \texttt{Comp} \} \}$, a typing environment $\Delta$ and operator typing environment $\Omega$, the judgment: $$\Delta, \Omega \imp  \Exe \{\void\ \main()\{ \texttt{Init}\ \texttt{Comp} \}\}  :\diamond$$ means that the program is well-typed with respect to $\Delta$ and $\Omega$ and is defined by: 
$$
\ninfer{(\void \ \main^{\Exe}(),\Delta), \Omega \Imp \texttt{Init}  : \void \quad (\void \ \main^{\Exe}(),\Delta), \Omega \imp\texttt{Comp}  : \void(\tierb) } 
{(\void \ \main^{\Exe}(),\Delta), \Omega \imp \Exe \{\void\ \main()\{ \texttt{Init}\ \texttt{Comp} \}\} : \diamond}
{}
$$
where $\Imp$ is a judgment derived from the type system by removing all tiers and tier based constraints in the typing rules. For example, Rule \textit{(C)} of Figure~\ref{fig:TypeE} 
becomes:
$$
\ninfer{
  \typenv,\Omega \Imp \ea : \C \quad \typenv,\Omega \Imp \overline{\ea} : \overline{\tau} \quad \typenv\{\tau\ \m^{\C}(\overline{\tau})\},\Omega\Imp  \tau\ \m^{\C}(\overline{\tau}) : \C \times \langle \overline{\tau} \rangle\to  \tau
 } 
{\typenv,\Omega\Imp  \ea.\m(\overline{\ea}):\tau}
{\textit{(C)}}
$$
or Rule \textit{(Body)} of Figure~\ref{fig:TypeCM} becomes:
$$
\ninfer{ \typenv\{\this\leftarrow \C, \overline{\xa} \leftarrow \overline{\tau},  [\xa \leftarrow \tau] \}, \Omega \Imp I : \void \quad \typenv,\Omega \Imp \this : \C }
{\typenv, \Omega \Imp \tau\ \m^{\C}(\overline{\tau}) : \C \times \langle \overline{\tau}\rangle \to  \tau}
{\textit{(Body)}} 
$$
It means that all the environments are lifted so that types are substituted to tiered types.

 Since no tier constraint is checked in the initialization instruction $\texttt{Init}$, the complexity of this latter instruction is not under control ;  as explained previously the main reason for this choice is that this instruction is considered to be building the program input. In contrast, the computational instruction $\texttt{Comp}$ is considered to be the computational part of the program and has to respect the tiering discipline.

\subsection{Examples}
\subsubsection{Well-typedness and type derivations}
\begin{example}\label{ex7}
Consider a program having the following computational instruction:
{\tt
\begin{lstlisting}
  Exe {
    void main(){
      //Init
      int n :=... ;
      BList b := null;
      while(n>0){
        b := new BList(true,b);
        n := n-1;
      }
      //Comp
      int z := 0;
      while(b.getTail() != null){
        b := b.getTail();
        z := z+1;
      }
    }
  }
\end{lstlisting}
}

It is well-typed as it can be typed by the following derivation:
$$
\ninfer{\ninfer{\ninfer{\typenv(\xc)=\ent(\tiera)}{\typenv, \Omega\imp_{\tiera} \xc: \ent(\tiera)}{(Var)}  \  \ninfer{}{\typenv, \Omega\imp_\tiera 0 :\ent(\tiera)}{(Cst)}} 
{\typenv, \Omega\imp \ \xc \iasg 0; : \void(\tiera)}
{\textit{(Ass)}}  \ \ninfer{\stackrel{\vdots}{(\Pi)} } 
{\typenv, \Omega\imp \ca :\void(\tierb)}
{(Wh)}  } 
{\typenv, \Omega\imp{\tt Comp} : \void(\tierb)}
{\textit{(Seq)}}
$$
where $\ca= \iwhile(\verb!b.getQueue()! != \nul)\{ \verb!b := b.getQueue(); z := z+1;!\}$.
Now we consider the sub-derivation $\Pi$, where we omit $\typenv,\Omega$ in the context in order to lighten the notation and where $\cb= \verb!b := b.getQueue(); z := z+1;!$:

$$
\ninfer
{
	\ninfer{
		\ninfer{\stackrel{\vdots}{\Pi_1}}
		{\imp_\tierb \tt b.getQueue() : BList(\tierb)}
		{(C)}
		\ninfer{}
		{\imp_\tierb \nul : \tt BList(\tierb)}
		{}
	}
	{\imp_\tierb \tt b.getQueue() != \nul : \bool(\tierb)}
	{}
       \ninfer{\stackrel{\vdots}{\Pi_2}}
       {\imp \cb  : \void(\tierb)}
       {(Seq)} 
}
{\imp \ca :\void(\tierb)}
{(Wh)} 
$$

For $\Pi_1$ we have:
$$ 
\ninfer{ 
	\ninfer{\typenv(\tt b)= {\tt BList}(\tierb)}
	{\imp_\tierb \tt b : {\tt BList}(\tierb)}
	{}
	 \ninfer{
	 	\ninfer{}
	 	{\{\tt getQueue\}[\this,\tt queue : {\tt BList}(\tierb)] , \Omega \imp ; : \void(\tierb) \ \cdots}
	 	{}
	}
	{\{\tt getQueue\}\imp_\tierb  \tt getQueue : BList(\tierb) \to  BList({\tierb})}
	{}
 } 
{\imp_{\tierb}  \tt b.getQueue():BList(\tierb)}
{\textit{(C)}}
$$
where $\{\tt getQueue \}$ is a shorthand notation for $\typenv\{{\tt BList}\ {\tt getQueue}^{\tt BList}() \} , \Omega$ and $[\this,\tt queue : {\tt BList}(\tierb)]$ is a shorthand notation for $[\this\leftarrow {\tt BList}(\tierb),  {\tt queue} \leftarrow {\tt BList}(\tierb)]$.

For $\Pi_2$, provided that $\mea=\tt b.getQueue() $, we have:
$$
 \ninfer{\
 	\ninfer{
 		\ninfer{\typenv(\tt b)=BList(\tierb)}
 		{\imp_\tiera \tt b : BList(\tierb)}
 		{}
 		 \ 
 		 \ninfer{\Pi_1}
 		 {\imp_\tierb \mea : \tt BList(\tierb)}
 		 {}
	}
 	{\imp \tt b :=\mea; :\void(\tierb)}
 	{}
 	\ 
 	\ninfer{
 		\ninfer{\typenv(\tt z)=\ent(\tiera)}
 		{\imp_\tiera \tt z : \ent(\tiera)}
 		{}
 		 \  
 		 \ninfer{\vdots}
 		 {\imp_\tiera \tt z+1 : \ent(\tiera)}
 		 {}
 	}
 	{\imp \tt z := z+1; : \void(\tiera)}
 	{}	
 }
{\imp \cb  : \void(\tierb)}
{(Seq)} 
$$
Though incomplete, the above derivation can be fully typed as we have already seen in Example~\ref{ex6} that $+1$ is a positive operator. Consequently, it can be given the type $\ent(\tiera) \to \ent(\tiera)$ and $\tt z +1$ can be typed, provided that $\typenv(\xc)=\ent(\tiera)$.
\end{example}

\subsubsection{Light notation: {\tt BList} revisited}
As we have seen above, type derivation can be tricky to provide. In order to overcome this problem, we will no longer use typing derivations: tiers will be made explicit in the code. The notation $\xa^\alpha$ means that $\xa$ has tier $\alpha$ under the considered environments $\typenv,\Omega$, \textit{i.e.} $\typenv,\Omega \imp_\beta \xa : \tau(\alpha)$, for some $\tau$ and $\beta$, whereas $\ca : \alpha$ means that $\typenv , \Omega \imp \ca : \void(\alpha)$.

\begin{example}\label{ex8}
We will consider some of the constructors and methods of the class $\tt BList$ to illustrate this notation.

{\tt
\begin{lstlisting}
BList { /* List of booleans: coding binary integers */
	boolean value;
	BList queue;

	BList(boolean v, BList q) {
		value := v;
		queue := q;
	}

	BList getQueue() { return queue; }

	void setQueue(BList q) {
		queue := q;
	}

	boolean getValue() { return value; }

	BList clone() {
		BList n;
		if (value != null) {
			n = new BList(value, queue.clone());
		} else {
			n = new BList(null, null);
		}
		return n;
	}

	void decrement() {
		if (value == true) {
			value := false;
		} else {
			if (queue != null) {
				value := true;
				queue.decrement();
			} else {
				value := false;
			}
		}
	}

	int length () {
		int res := 1; 
		if ( queue != null ) {
			res := queue.length();
			res := res +1;
		}
		else {;}
		return res ;
	}

	boolean isEqual(BList other) {
		boolean res := true;
		BList b1 := this;
		BList b2 := other;
		while (b1 != null && b2 != null) {
			if (b1.getValue() != b2.getValue()){
				res := false
			} else {;;}
			b1 := b1.getQueue();
			b2 := b2.getQueue();
		}
		if (b1 != null || b2 != null) {
			res := false;
		}
		return res;
	}
}
\end{lstlisting}
}

Let us now consider and type each method and constructor:

{\tt
\begin{lstlisting}
    BList(boolean v$^\tiera$, BList q$^\tiera$) {
      value := v; : $\tiera$
      queue := q; : $\tiera$
    }
\end{lstlisting}
}
The constructor $\tt BList$ can be typed by $\bool(\tiera) \times \tt Blist(\tiera) \to BList(\tiera)$,
using Rules \textit{(Cons)} and \textit{(Seq)} and twice Rule \textit{(Ass)}. In this latter applications, the tier are enforced to be $\tiera$ because of the side condition $\xa \in \mathcal{F} \implies \alpha=\tiera$. As a consequence, it is not possible to create objects of tiered type $\tt BList(\tierb)$ in a computational instruction.
{\tt
\begin{lstlisting}
    BList getQueue() {
      return queue;
    }
\end{lstlisting}
}
We have already seen in Example~\ref{ex7} that the method  {\tt getQueue} can be typed by $\tt BList(\tierb) \to BList(\tierb)$. It can also be typed by $\tt BList(\tiera) \to BList(\tiera)$, by Rules \textit{(Body)} and \textit{(Self)}.

The types $\tt BList(\alpha) \to BList(\beta)$, $\alpha \neq \beta$, are prohibited because of Rules \textit{(Body)} and \textit{(Self)} since the tier of the current object has to match the tier of its fields.

In the same manner, the method {\tt getValue} can be given the types $\tt BList(\alpha) \to \bool(\alpha)$, $\alpha \in \{\tiera,\tierb\}$.

The method {\tt setQueue}:
{\tt
\begin{lstlisting}
    void setQueue(BList q$^\tiera$) {
      queue := q; : $\tiera$
    }
\end{lstlisting}
}
can only be given the types $\tt BList(\tiera) \times BList(\tiera) \to \void(\beta)$. Indeed, by Rule \textit{(Ass)}, $\tt queue$ and $\tt q$ are enforced to be of tier $\tiera$. Consequently, $\this$ is of tier $\tiera$ by Rules \textit{(Body)} and \textit{(Self)}. The tier of the body is not constrained.

Consider the method $\tt decrement$:
{\tt
\begin{lstlisting}
    void decrement() { 
      if (value$^\tiera$ == true) {
        value$^\tiera$ := false; $: \tiera$
      } else {
        if (queue$^\tiera$ != null) {
          value := true;
          queue$^\tiera$.decrement(); $: \tiera$
        } else {
          value$^\tiera$ := false; $: \tiera$
        }
      }
    } $: \tiera$ 
\end{lstlisting}
}
Because of the Rules \textit{(Ass)} and \textit{(Body)}, it can be given the type $\tt BList(\tiera) \to \void(\tiera)$. As we shall see later in Subsection~\ref{recmet}, this method will be rejected by the safety condition as it might lead to exponential length derivation whenever it is called in a while loop.

Now consider the following method:
{\tt
\begin{lstlisting}
  int length() {
    int res := 1; : $\tiera$
    if (queue$^\tierb$ != null) {
      res := queue.length(); : $\tierb$ //using (Isub)
      res := res+1; : $\tiera$
    }
    else {;}
    return res;
  }
\end{lstlisting} 
}
It can be typed by $\Gamma,\Omega \imp_\tierb \tt int\ length^{\tt BList}() :  BList(\tierb) \to \ent(\tiera)$.

Now consider the method testing the equality:
{\tt
\begin{lstlisting}
    boolean isEqual(BList other$^\tierb$) {
      boolean res$^\tiera$ := true; $: \tierb$
      BList b1$^\tierb$ := this$^\tierb$; $: \tierb$
      BList b2$^\tierb$ := other$^\tierb$; $: \tierb$
      while (b1$^\tierb$ != null && b2$^\tierb$ != null){
         if(b1$^\tierb$.getValue() != b2$^\tierb$.getValue()){
           res$^\tiera$ := false;$: \tierb$ //using (ISub)
         } else {;}
         b1$^\tierb$ := b1$^\tierb$.getQueue();$: \tierb$
         b2$^\tierb$ := b2$^\tierb$.getQueue();$: \tierb$
      }
      if (b1$^\tierb$ != null || b2$^\tierb$ != null) {
        res$^\tiera$ := false;$: \tierb$ //using (ISub)
      } else {;} $: \tierb$  
      return res$^\tiera$;
    }
\end{lstlisting}
}
The local variables $\tt b1$ and $\tt b2$ are enforced to be of tier $\tierb$ by Rules \textit{(Wh)} and (Op). Consequently, $\this$ and $\tt other$ are also of tier $\tierb$ using twice Rule (Ass). This is possible as it does not correspond to field assignments. Moreover, the methods $\tt getValue$ and $\tt getQueue$ will be typed by $\tt BList(\tierb) \to boolean(\tierb)$ and $\tt BList(\tierb) \to BList(\tierb)$, respectively. Finally, the local variable can be given the type $\bool(\tierb)$ or $\bool(\tiera)$ (in this latter case, the subtyping Rule (ISub) will be needed as illustrated in the above instruction) and, consequently, the admissible types for $\tt isEqual$ are $\tt BList(\tierb) \times BList(\tierb) \to \bool(\alpha)$, $\alpha \in \{\tiera,\tierb\}$.
\end{example}

\subsubsection{Examples of overriding}
\begin{example}[Override1]
Consider the following classes:
{\tt
\begin{lstlisting}
A {
  int x;
   ...
  void f(int y){  
    x := x+1; : $\tierb$ //using (ISub)
  }
}

B extends A{
  void f(int y){
    while(y$^\tierb$>0){
      x := x+1; : $\tiera$
      y := y-1; : $\tierb$
    }
  }
}
\end{lstlisting}
}
In the class $\verb!A!$, the method $\verb!f!$ can be given the type $\verb!A!(\tiera)\times \ent(\tierb) \to \void(\tierb)$ using Rules \textit{(Op)}, \textit{(Ass)}, \textit{(Self)} and \textit{(ISub)}.
In the class $\verb!B!$, the method $\verb!f!$ can be given the type $\verb!B!(\tiera)\times \ent(\tierb) \to \void(\tierb)$ using Rules \textit{(Op)}, \textit{(Ass)}, \textit{(Seq)}, \textit{(Wh)} and \textit{(Self)}.
Consequently, the following code can be typed:
{\tt
\begin{lstlisting}
A x = ...
if (condition){
  x=new A(...);
}else{
  x=new B(...);
}
x.f(25);
\end{lstlisting}
}
provided that $\tt condition$ is of tier $\tierb$ and using Rules \textit{(Pol)} and \textit{(OR)}. Indeed, we are not able to predict statically which one of the two method will be called. However we know that the argument is of tier $\tiera$ as its field will increase polynomially in both cases.
\end{example}

\begin{example}[Override2]
Consider the class $\tt BList$ and a subclass $\tt PairOfBList$ :
{\tt
\begin{lstlisting}
public class BList{
	...
	int length () {
		int res := 1; 
		if ( queue != null ) {
			res := queue.length();
			res := res +1;
		}
		else {;}
		return res ;
	}
}

public class PairOfBList extends BList {
	Blist l$^\tierb$;

	int length () {
		int res$^\tiera$ := queue$^\tierb$.length()$^\tiera$+2;
		while ( l$^\tierb$ != null ) {
			l$^\tierb$ := l$^\tierb$.getQueue();
			res$^\tiera$ := res$^\tiera$ +1;
		}
		else {;}
		return res$^\tiera$ ;
	}
}
\end{lstlisting}
}
The override method $\tt length$ of $\tt PairOfBList$ computes the size of a pair of $\tt BList$ objects, that is the sum of their respective size. As highlighted by the annotations, it can be given the type $\tt PairOfBList(\tierb) \to \ent(\tiera)$. Consequently, a method call of the shape:
$$\tt \ent \ i^\tiera := p.length()^\tiera; : \void(\tierb)$$
$\tt p$ being of type $\tt PairOfBList$, can be typed using the Rule (OR):
$$
\ninfer{\tt PairOfBList \dord BList \qquad \typenv,\Omega \imp_\tierb \ent\ length^{BList}() : BList(\tierb) \to  \ent({\tiera}) }{\tt \typenv,\Omega\imp_\tierb \ent\ length^{PairOfBList}() : PairOfBList({\tierb})\times \to  \ent({\tiera})}
{}
$$
as Example~\ref{ex8} has demonstrated that the judgment $\typenv,\Omega \imp_\tierb \ent\ length^{BList}() : BList(\tierb) \to  \ent({\tiera})$ can be derived. In order for the program to be safe, the tier annotation is $\tierb$ and, consequently, the assignment is of tiered type $\void(\tierb)$ (hence cannot be guarded by $\tiera$ data).
\end{example}

\subsubsection{Exponential as a counter-example}
\begin{example}\label{ex9}
Now we illustrate the limits of the type system with respect to the following method computing an exponential:
{\tt
\begin{lstlisting}
  exp(int x, int y){
    while(x$^\tierb$>0){ 
        int u$^?$ := y$^\tiera$;
        while (u$^?$>0){
           y := y$^\tiera$+1; : $\tiera$
           u := u$^?$-1; : $?$
        }
        x := x$^\tierb$-1;
    }
    return y$^\tiera$;
  }
\end{lstlisting}
}
It is not typable in the presented formalism. Indeed, suppose that it is typable. The expression $\verb!y+1;!$ enforces $\xb$ to be of tier $\tiera$ by Rule \textit{(Op)} and by definition of positive operators. Consequently, the instruction $\verb!int u := y;!$ enforces $\xd$ to be of tier $\tiera$ because of typing discipline for assignments (Rule \textit{(Ass)}). However, the innermost while loop enforces $\xd >0$ to be of tier $\tierb$ by Rules \textit{(Wh)} and \textit{(Op)}, so that $\xd$ has to be of tier $\tierb$ and we obtain a contradiction.

Now, one might suggest that the exponential could be computed using an intermediate method $\verb!add!$ for addition:
{\tt
\begin{lstlisting}
add(int x, int y){                       
  while(x$^\tierb$>0){
    x$^\tierb$ := x$^\tierb$-1; $:\tierb$
    y$^\tiera$ := y$^\tiera$+1; $:\tiera$
  }
  return y$^\tiera$;
}                             

expo(int x){
  int res := 1;
  while(x$^\tierb$>0){ 
      res := add(res$^?$, res$^?$);
      x := x$^\tierb$-1; : $\tierb$
  }
  return res;
}
\end{lstlisting}
}
Thankfully, this program is not typable as the first argument of $\verb!add!$ is enforced to be of tier $\tierb$ and the second argument is enforced to be of tier $\tiera$ (see previous section). Hence, the variable $\verb!res!$ would have two distinct tiers under the same context. That is clearly not allowed by the type system.
\end{example}

\subsection{Type preservation under flattening}

We show that the flattening of a typable instruction has a type preservation property. A direct consequence is that the flattened program can be considered instead of the initial program.
\begin{proposition}\label{lemma:typepres}
Given an instruction $\ca$, a contextual typing environment $\typenv$ and an operator typing environment $\Omega$ such that $\typenv,\Omega \imp \ca : \void(\alpha)$ holds, there is a contextual typing environment $\typenv'$ such that the following holds:
\begin{itemize}
\item $\forall \xa \in \ca,\ \typenv'(\xa)=\typenv(\xa)$
\item $\typenv',\Omega \imp \underline{\ca} : \void(\alpha)$
\end{itemize}
where $\xa \in \ca$ means that the variable $\xa$ appears in $\ca$.

Conversely, if $\typenv',\Omega \imp \underline{\ca} : \void(\alpha)$, then $\typenv', \Omega \imp {\ca} : \void(\alpha)$.
\end{proposition}

\begin{proof}
By induction on program flattening on instructions. Consider a method call $\ca =\tau \  \xa \iasg \ea.\m(\ea_1,\ldots,\ea_n);$ such that $\typenv,\Omega \imp \ca : \void(\alpha)$ and $\typenv =(s,\Delta)$. This means that  $\typenv,\Omega \imp_{\gamma_i} \ea_i : \tau_i(\alpha_i)$, $\typenv,\Omega \imp_{\gamma} \xa : \tau(\alpha)$ and  $\typenv,\Omega \imp_{\gamma'} \ea : \C(\beta)$ hold, for some $\gamma_i$, $\gamma$, $\gamma'$, $\tau_i$, $\alpha_i$, $\tau$, $\alpha$, $\C$ and $\beta$ (see Rule \textit{(C)} of Figure~\ref{fig:TypeE}). Applying the rules of Figure~\ref{flat}, the flattening of $\ca$ is of the shape $\underline{\cb}\ [\tau]\ \xa =\xa_{n+1}.\m'(\xaa,\ldots,\xan);$ with $\cb = \tau_1\ \xaa \iasg \eaa; \ldots \tau_n\ \xan \iasg \ean; \tau_{n+1}\ \xa_{n+1} \iasg \ea;$. Let $\typenv'$ for the environment that is equal to $\typenv$ but on the method signature $s$ where:
$$\typenv'(\xb)=\begin{cases} 
\tau(\alpha) \text{ if } \xb =\xa\\
\tau_i(\alpha_i) \text{ if } \xb \in \{\xa_1,\ldots, \xa_n\}\\
\C'(\beta) \text{ if } \xb =\xa_{n+1}\\
\typenv(\xb) \text{ otherwise }
\end{cases}
$$
We have
$\typenv', \Omega \imp  J \ [\tau]\ \xa =\xa_{n+1}.\m'(\xaa,\ldots,\xan);:\void(\alpha)$ and $\typenv' ,\Omega \imp J  : \void(\alpha)$ (sub-typing might be used). By induction hypothesis, there is a contextual typing environment $\typenv''$ such that $\typenv'',\Omega \imp \underline{J} : \void(\alpha)$ and $\forall\xa \in \ca,\ \typenv''(\xa)=\typenv'(\xa)=\typenv(\xa)$ and, consequently, $\typenv'' \imp_\gamma \underline{I}: \void(\alpha)$. All the other cases are treated similarly.

The converse is straightforward as $\typenv$ and $\typenv'$ match on the variables that they ``share'' in common.
\end{proof}

\subsection{Subject reduction}
We introduce an intermediate lemma stating that an instruction obtained through the evaluation of the computational instruction of a well-typed program is also well-typed. For this, we first need to extend the type system so that it will be defined on any meta-instruction:
\begin{figure}[!ht]
$$
\ninfer{} 
{\typenv, \Omega\imp  \pop; :\void(\alpha)}
{\textit{(Pop)}}  
\qquad
\ninfer{} 
{\typenv, \Omega\imp  \push(s_\heap); :\void(\tiera)}
{\textit{(Push)}}  
$$
\caption{Extra typing rules for Push and Pop instructions} \label{pp}
\end{figure}

Notice that this definition is quite natural as $\push$ makes the memory increase and so is of tier $\tiera$ while $\pop$ makes the memory decrease and so can be of any tier. One possibility would have been to include these typing rules directly in the initial type system. We did the current choice as these meta-instructions are only obtained during computation whereas the type inference is supposed to be performed statically on the flattened program, that is an AOO program, independently of its execution.

\begin{lemma}[Subject reduction]\label{lem:subred}
Let $\typenv$ be a contextual typing environment, $\Omega$ be an operator typing environment and $P$ be an AOO program of executable $\Exe\{ \void\ \main()\{ {\tt Init}\ {\tt Comp} \} \}$ such that $P$ is well-typed with respect to $\typenv$ and $\Omega$. If  $(\mathcal{I},\underline{\tt Comp}) \to^* (\conf, \mi)$ then $\typenv, \Omega \imp \mi : \void(\tierb)$.
\end{lemma}
\begin{proof}
By Proposition~\ref{lemma:typepres}, $\underline{\tt Comp}$ can be typed.
We proceed by induction on the reduction length $n$ of $\to$ ; $\to^n$ being a notation for a length $n$ reduction:
\begin{itemize}
\item If $n=0$ then $\mi=\underline{\tt Comp}$ and since the program is well-typed, we have $\typenv, \Omega \imp \underline{\tt Comp} : \void(\tierb)$.
\item Now consider a reduction of length $n+1$. It can be written as $(\mathcal{I},\underline{\tt Comp}) \to^n (\conf,\mi) \to (\conf',\mi')$. By induction hypothesis $\typenv, \Omega \imp \mi : \void(\tierb)$. Moreover, the last rule can be:
\begin{itemize}
\item the evaluation of an assignment (Rules (1-7) of Figure~\ref{fig:sem}) or the evaluation of a push or pop instructions (Rules (9) and (10)  of Figure~\ref{fig:sem}). In all these cases, $\mi=\mi_1 \ \mi'$ and $(\conf,\mi) \to (\conf', \mi')$. By Rule \textit{(Seq)} of Figure~\ref{fig:TypeI}, $\typenv, \Omega \imp \mi' : \void(\alpha)$, for some $\alpha$ and, consequently, $\typenv, \Omega \imp \mi' : \void(\tierb)$ using Rule \textit{(ISub)} of Figure~\ref{fig:TypeI}.
\item the evaluation of a method call (Rule (8) of Figure~\ref{fig:sem}). In this case, $(\conf,\mi)=(\conf,[\xa \iasg]\xc.\m(\overline{\xb});\mi') \to (\conf', \push(s_\heap);\mi'' \ [\xa \iasg \xc';]\ \pop;\ \mi )$, provided that $\mi''$ is the flattened body of method $\m$. By Rule \textit{(Seq)} of Figure~\ref{fig:TypeI}, $\typenv, \Omega \imp \mi' : \void(\alpha)$, for some $\alpha$. By the extra rules of Figure~\ref{pp}, $\typenv, \Omega \imp \push(s_\heap); : \void(\tiera)$ and $\typenv, \Omega \imp \pop; : \void(\tiera)$. Moreover by Rules \textit{(Body)} of Figure~\ref{fig:TypeI} and \textit{(C)} of Figure~\ref{fig:TypeE}, the flattened body $\mi''$ can be typed by $\typenv, \Omega \imp \mi'': \void(\beta)$, for some $\beta$ (the method output tier). We let the reader check that $\xa \iasg \xc';$ can also be typed using the same reasoning in the case where the method returns something. Finally, using several times Rule \textit{(Seq)} and one time Rule \textit{(ISub)} of Figure~\ref{fig:TypeI}, we obtain that $\typenv, \Omega \imp \push(s_\heap);\mi'' \ [\xa \iasg \xc';]\ \pop;\ \mi : \void(\tierb)$.
\item the evaluation of a while loop (Rules (11) and (12) of Figure~\ref{fig:sem}). In this case $\mi =  \iwh(\xa)\ido \{ \mi_1  \}\ \mi_2$  and either $\xa$ is evaluated to false, in which case we take $\mi' = \mi_2$  or $\xa$ evaluates to true and $\mi'= \mi_2\ \iwh(\xa)\ido \{ \mi_1  \}$. In both cases, $\typenv, \Omega \imp  \mi_2 : \void(\alpha)$, some $\alpha$, and $\typenv, \Omega \imp \iwh(\xa)\ido \{ \mi_1  \} : \void(\tierb)$ can be derived by Rules $\textit{(Seq)}$ and $\textit{(Wh)}$  of Figure~\ref{fig:TypeI}. Consequently, we can derive $\typenv, \Omega \imp  \mi' : \void(\tierb)$ by either using Rule $\textit{(ISub)}$ (for false), if needed, or just by using Rule $\textit{(Seq)}$ (for true).
\item the evaluation of a {\tt if} ((Rule (13) of Figure~\ref{fig:sem}) can be done in a similar manner.
\end{itemize}
and so the result.\qedhere
\end{itemize}
\end{proof}

\begin{remark}
Subject reduction is presented in a weak form as we do not consider the tier $\tiera$ case. It also holds but is of no interest in this particular case as, by monotonicity of the typing rule, every sub-instruction will be of tier 0 and there will be no loop and no recursive call. This is the reason why the computational instruction is typed using tier $\tierb$ in the definition of well-typedness.
\end{remark}

\section{Safe recursion}
\subsection{Recursive methods}
Given two methods of signatures $s$ and $s'$ and names $\m$ and $\m'$, define the relation $\mcall$ on method signatures by $s\ \mcall s'$ if $\m'$ is called in in the body of $\m$.
This relation is extended to inheritance 
by considering that overriding methods are called by the overridden method.
Let  $\mcall^+$ be its transitive closure.
A method of signature $s$ is \emph{recursive} if $s\ \mcall^+ s$ holds. Given two method signatures $s$ and $s',\ s \equiv s'$ holds if both $s \mcall^+ s'$ and $s' \mcall^+ s$ hold. Given a signature $s$, the class $[s]$ is defined as usual by $[s]=\{s' \ | \  s' \equiv s \}$. Finally, we write $s \smcall^+ s'$ if  $s \mcall^+ s'$ holds but not $s' \mcall^+ s$.

Notice that the extension of $\mcall$ to method override is rough as we consider overriding methods to be called by the overridden method though it is clearly not always the case. However it is put in order to make the set $[s]$ computable for any method signature $s$. Indeed in a method call of the shape $\xa.\m()$, provided that $\xa$ is a declared variable of type $\C$ and that $\void\ \m()$ is a method declared in $\C$ and overridden in some subclass $\D$, the evaluation may lead dynamically to either a call to  $\void\ \m^\C()$ or to a call to $\void\ \m^\D()$ depending on the instance that will be assigned to $\xa$ during the program evaluation. This choice is highly undecidable as it depends on program semantics. Here we will consider that $\void\ \m^\C() \mcall \void\  \m^\D()$. Consequently, if  the call $\xa.\m()$ is performed in the body of a method of signature $s$ then $s\ \mcall\ \void\ \m^\C()\ \mcall\ \void \ \m^\D()$ and, consequently, $s\ \mcall^+ \ \void\ \m^\D()$, \textit{i.e.} a call to the method of signature $s$ may lead to a call to the method of signature $\void\ \m^\D()$. 

\begin{lemma}\label{lemma:lev}
Given an AOO program, for any method signature $s$, the set $[s]$ is computable in polynomial time in the size of the program.
\end{lemma}
\begin{proof}
It just suffices to look at the program syntax to generate linearly the relation $\mcall$ and then to compute its transitive closure $\mcall^+$.
\end{proof}

\subsection{Level}

The notion of level of a meta-instruction is introduced to compute an upper bound on the number of nested recursive calls for a method call evaluation.

\begin{definition}[Level]
The level $\lev$ of a method signature is defined by:
\begin{itemize}
\item 
$\lev(s)=\max \{ \lev(s') \ | \ s \mcall s'\}$ if $s \notin [s]$ (\textit{i.e.} the method is not recursive),
\item
$\lev(s)=1+ \max \{ \lev(s') \ | \ s \smcall^+ s'\}$ otherwise,
\end{itemize}
setting $\max(\emptyset)=0$.

Let $\lev(P)$ be  the maximal level of a method in a program $P$ (or $\underline{P}$).  By abuse of notation, we will write $\lev(\m)$ and, respectively, $\lev$ when the signature of the method $\m$ and the program $P$, respectively, are clear from the context.
\end{definition}

\begin{example}\label{ex10}
Consider the methods of Example~\ref{ex8}:
\begin{itemize}
\item $\tt \lev(getQueue) = \lev(getValue)= \lev(setQueue) = 0$ as these methods do not call any other methods in their body and, consequently, are not recursive,
\item $\tt \lev(isEqual)=0$ as $\tt isEqual$ call the methods $\tt getQueue$ and $\tt getValue$ but the converse does not hold. Consequently, $\tt BList \ isEqual^{\tt BList}(BList)  \smcall^+ BList \ getQueue^{\tt BList}()$ and $[\tt BList \ isEqual^{\tt BList}(BList)] = \emptyset$,
\item $\tt \lev(decrement)=1+ \max \{ \lev(s') \ | \ \void \ decrement^{\tt BList}() \smcall^+ s'\} = 1+\max(\emptyset) =1,$
\item $\tt \lev(length)=1$, for the same reason.
 \end{itemize}
\end{example}

\subsection{Intricacy}
The notion of intricacy corresponds to the number of nested $\iwh$ loops in a
meta-instruction and will be used to compute the requested upper bounds.

\begin{definition}[Intricacy]
Let the intricacy $\nest$ be partial function from meta-instructions to integers defined as follows:
\begin{itemize}
\item ${\nest}([[\tau] \xa \iasg] \mea;) = 0$ if $\mea$ is not a method call, 
\item ${\nest}([{[\tau]}\ \xa \iasg] \xb.\m(\ldots);) =\max_{\D \dord \C\star}(\nu(\mi_{\D})) $\\
provided that $\m$ is of the shape $\tau\ \m(\ldots)\{\mi_{\D}\ [\iret\ \xc;]\}$ in a class $\D \dord \C\star$, where $\xb$ is of type $\C$ and $\C \star$ is the least super-class of $\C$ where $\m$ is defined.
\item ${\nest}(\mi\ \mi') = \max({\nest}(\mi), {\nest}(\mi'))$
\item ${\nest}(\instr{if}(\xa) \{\mi\}\instr{else}\{\mi'\}) = \max({\nest}(\mi), {\nest}(\mi'))$
\item ${\nest}(\iwh (\xa)\{\mi\}) = 1 + {\nest}(\mi)$
\end{itemize}
 Moreover, let $\nest(P)$ be the maximal intricacy of a meta-instruction within the flattened AOO program $\underline{P}$. By abuse of notation, we will write $\nest$ when the program $P$ is clear from the context.
\end{definition}
Observe that for any instruction $\ca$, the intricacy of its flattening $\nest(\underline{\ca})$ is well-defined as there is no $\push$ or $\pop$ operations occuring in it. Moreover, in the simple case where there is no inheritance then the intricacy of a method call ${\nest}([{[\tau]}\ \xa \iasg] \xb.\m(\ldots);)$, for some method of the shape $\tau\ \m(\ldots)\{\mi\ [\iret\ \xc;]\}$ is just equal to $\nest(\mi)$ (as $\C \star =\C$ and the only $\D$ such that $\D \dord \C$ is $\C$ itself).

\begin{example}
Consider the following meta-instruction $\mi$:
{\tt
\begin{lstlisting}
   while(x){
     while(y){
       b := l.isEqual(o);
     }
   }
\end{lstlisting}
}
$\nest(\mi)=2+\nest(\mi')$, if $\mi'$ is the flattened body of the method $\tt isEqual$. We let the reader check that $\nest(\mi')=1$ (there is one while inside). Consequently, $\nest(\mi)=3$.
\end{example}

\subsection{Safety}\label{recmet}
Now we put some aside restrictions on recursive methods to ensure that their computations remain polynomially bounded.

\begin{definition}[Safety]\label{def:safety}
A well-typed program with respect to a typing environment $\Delta$ and operator typing environment $\Omega$ is \emph{safe} if for each recursive method  $\tau \ \m(\overline{\tau\ \xa})\{ \ca\ [\return\ \xa;]\}$:
\begin{enumerate}
\item there is exactly one call to some $\m' \in [\m]$ in the instruction $\ca$,\label{D1}
\item there is no while loop inside $\ca$, i.e. $\nu(\ca)=0$,\label{D2}
\item and only judgments of the shape $(s, \Delta),\Omega \imp_\tierb \tau \ \m^\C(\overline{\tau\ \xa}) :\C(\tierb)\times \langle \overline{\tau}(\tierb) \rangle \to \tau(\alpha)$ can be derived using Rules $\textit{(Body)}$ and $\textit{(OR)}$.\label{D3}
\end{enumerate}
\end{definition}
Item~\ref{D1} ensures that recursive methods will be restricted to have only one recursive call performed during the evaluation of their body. It will prevent exponential accumulation through mutual recursion. A counter-example is a program computing the Fibonacci sequence with a mutually recursive call of the shape $\m(n-1)+\m(n-2)$, for $n>2$.
Item~\ref{D2} is here to simplify the complexity analysis. It is not that restrictive in the sense that it enforces the programmer to make a choice between a functional programming style using pure recursive calls or an imperative one using while loops. 
Item~\ref{D3} is very important. It enforces recursive methods to have tier $\tierb$ inputs in order to prevent a control flow depending on tier $\tiera$ variables during the recursive calls evaluation and to have an output whose use is restricted to tier $\tierb$ instructions (this is the purpose of the $\imp_\tierb$ use) in order to prevent the recursive method call to be controlled by tier $\tiera$ instructions.
The method output tier $\alpha$ is not restricted.

\begin{example}\label{11}
A well-typed program whose computational instruction uses the methods $\tt getQueue,\ getValue,\ setQueue,\ length$ and $\tt isEqual$ of Example~\ref{ex8} will be safe. Indeed  the only recursive method is length. As illustrated in Example~\ref{ex8}, it can be typed by $\tt BList(\tierb) \to \ent(\tiera)$, it does not contain any while loop and has only one recursive call in its body.

A program whose computational instruction uses the method $\tt decrement$ will not be safe as this method can only be given the type $BList(\tiera) \to \void(\tiera)$. Though it entails a lack of expressivity, there is no straightforward reason to reject this code. However we will see clearly in the next subsection that we do not want tier $\tiera$ data to control a while or a recursion while we do not want tier $\tierb$ data to be altered. Changing the type system by allowing $\tt decrement$ to apply to tier $\tierb$ objects would allow codes like:
{\tt
\begin{lstlisting}
   while(!o.isEqual(l)){
     o.decrement();
   }
\end{lstlisting}
}
where $\tt l$ is a $\tt BList$ of $n$ booleans equal to $\false$. Clearly, such a loop can be executed exponentially in the input size (the size of the lists). This is highly undesirable.
\end{example}

An important point to mention is that safety allows an easy way to perform expression subtyping for objects through the use of cloning. In other words, it is possible to bring a reference type expression from tier $\tierb$ to tier $\tiera$ based on the premise that it has been cloned in memory (this was already true for primitive data by Rule \textit{(Ass)}). Thus a form of subtyping that does not break the non-interference properties is admissible as illustrated by the following example:
\begin{example}
Let us add a clone method to the to the class $\tt BList$. This clone method will output a copy of the current object that, as it is constructed, will be of tier $\tiera$.

{\tt
\begin{lstlisting}
  BList clone(){
    BList res$^\tiera$ := null;
    int v$^\tiera$ := value$^\tierb$;
    if(queue$^\tierb$ == null){
      res$^\tiera$ := new BList(v$^\tiera$,null)$^\tiera$;
    }
    else{
     res$^\tiera$ := new BList(v$^\tiera$,queue.clone()$^\tiera$);
    }
    return res;
  }
\end{lstlisting}
}
This is typable by $\tt BList(\tierb) \to BList(\tiera)$. Moreover, the method is safe, as there is only one recursive call.
\end{example}

\begin{lemma}\label{lemma:safeequiv}
A program $P$ is safe iff $\underline{P}$ is safe. 
\end{lemma}
\begin{proof}
By Proposition~\ref{lemma:typepres}, $P$ is well-typed iff $\underline{P}$ is well-typed. The flattening transformation has no effect on method calls, while loops and method signatures. Consequently, $P$ satisfies Items~\ref{D1},~\ref{D2} and~\ref{D3} iff so does $\underline{P}$.
\end{proof}

\begin{lemma}\label{lemma:safecheck}
Given a well-typed program $P$ with respect to some typing derivation $\Pi$, the safety of $P$ can be checked in polynomial time in $\size{P}$. 
\end{lemma}
\begin{proof}
By Lemma~\ref{lemma:lev}, computing the sets of ``equivalent'' recursive methods can be done in polynomial time in $\size{P}$. Once, this is done, checking that there is no while loops and only one recursive call inside can be done quadratically in the program size. Finally, checking for Item~\ref{D3} can be done linearly in the size of the typing derivation $\Pi$ of $P$. However such a derivation is quadratic in the size of the program as all typing rules apart from $\textit{(ISub)}$ and $\textit{(Pol)}$ correspond to some program construct. Just notice that Rule $\textit{(ISub)}$ can be applied only once per instruction while Rule  $\textit{(Pol)}$ can be applied at most $k$ times per expression, provided that $k$ is an upper bound on the number of nested inheritance (and $k$ is bounded by $\size{P}$).
\end{proof}

\begin{corollary}
Given a well-typed program $P$ with respect to some typing derivation $\Pi$, the safety of $\underline{P}$ can be checked in polynomial time in $\size{P}$. 
\end{corollary}
\begin{proof}
By Lemma~\ref{lemma:safecheck}, the safety of $P$ can be checked in polynomial time in $\size{P}$ and so the result, by Lemma~\ref{lemma:safeequiv}. 
\end{proof}

\subsection{General safety}
The safety criterion is sometimes very restrictive from an expressivity point of view because of Item~\ref{D1}.
This Item is restrictive but simple: the interest is to have a decidable criterion for safety while keeping a completeness result for polynomial time as we shall see shortly. However it can be generalized to a semantics (thus undecidable criterion) that we call general safety in order to increase program expressivity.

\begin{definition}[General safety]\label{def:gsafety}
A well-typed program with respect to a typing environment $\Delta$ and operator typing environment $\Omega$ is \emph{generally safe} if for each recursive method  $\tau \ \m(\overline{\tau\ \xa})\{ \ca\ [\return\ \xa;]\}$:
\begin{enumerate}
\item the following condition is satisfied:
\begin{enumerate}
\item either there is at most one call to some $\m' \in [\m]$ in the evaluation of $\ca$,\label{c1}
\item or all recursive calls are performed on a distinct field of the current object. Such fields being used at most once.\label{c2}
\end{enumerate}
\item there is no while loop inside $\ca$, i.e. $\nu(\ca)=0$,\label{c4}
\item and only judgments of the shape $(s, \Delta),\Omega \imp_\tierb \tau \ \m^\C(\overline{\tau\ \xa}) :\C(\tierb)\times \langle \overline{\tau}(\tierb) \rangle \to \tau(\alpha)$ can be derived using Rules $\textit{(Body)}$ and $\textit{(OR)}$.\label{c3}
\end{enumerate}
\end{definition}

\begin{example}
General safety improves the expressivity of captured programs as
illustrated by the following example:

{\tt
\begin{lstlisting}
class Tree {
  int node; 
  Tree left;
  Tree right;
  
  int value(BList b$^\tierb$) {
    int res$^\tierb$ := node$^\tierb$;
    if(b$^\tierb$ != null && right$^\tierb$ != null && left$^\tierb$ != null){
      if(b.getValue()){
        res := right.value(b.getQueue()$^\tierb$); : $\tierb$
      }else{ 
        res := left.value(b.getQueue()$^\tierb$); : $\tierb$
      }
   }else{;} 
   return res;
  }
}
\end{lstlisting}
}

The method $\verb!value!$ returns the value of the node whose path is encoded by a boolean list in the method parameter $\verb!b!$ (the boolean constants $\verb!true!$ and $\verb!false!$ encoding the right and left sons, respectively). It can be typed by $\verb!Tree!$($\tierb$) $\times$ $\verb!BList!$($\tierb$) $\to$ $\ent$($\tierb$). Consequently, it is generally safe (but not safe) as exactly one branch of the conditional (thus one recursive call) can be reached, thus satisfying Item~\ref{c1} of Definition~\ref{def:gsafety}.

In the same manner, Item~\ref{c2} allows us to implement the clone method on recursive structures:
{\tt
\begin{lstlisting}
  Tree clone(){
    Tree res$^\tiera$ := null;
    int n$^\tiera$ := node$^\tierb$;
    if(left$^\tierb$ == null && right$^\tierb$ == null){
    //The case of non well-balanced trees is not treated. 
      res$^\tiera$ := new Tree(n$^\tiera$,null,null)$^\tiera$;
    }
    else{
     res$^\tiera$ := new Tree(n$^\tiera$,left.clone()$^\tiera$, right.clone()$^\tiera$);
    }
    return res;
  }
\end{lstlisting}
}
This is typable by $\tt Tree(\tierb) \to Tree(\tiera)$ thanks to Item~\ref{c2} of Definition~\ref{def:gsafety} as the two recursive calls are applied exactly once on distinct fields ($\tt left$ and $\tt right$).
\end{example}

We can now compare safety and general safety:
\begin{proposition}\label{afortiori}
If a program is safe then it is generally safe. The converse is not true.
\end{proposition}
\begin{proof}
If there is only one recursive call in the body of any recursive method then, a fortiori, there is only one recursive call during the body evaluation. On the opposite in an instruction of the shape $\ca=\iif (\ea)\ithen\{ \m();\}\ielse \{\m();\}$ there is only one call of $\m()$ during the evaluation of $\ca$ whereas $\m$ is called twice in $\ca$.
\end{proof}

\begin{proposition}\label{undec}
General safety is undecidable.
\end{proposition}
\begin{proof}
Since Item~\ref{c1} can be reduced to either a dead-code problem or a reachability problem that are known to be undecidable problems. Items~\ref{c2},~\ref{c4} and~\ref{c3} are still decidable in polynomial time.
\end{proof}

Note that there are decidable classes of programs between safe programs and generally safe programs. For example, one may ask a program to have only one recursive call per reachable branch of a conditional in a method body. This requirement is clearly decidable in polynomial time.

Hence general safety can be seen as a generalization of the safe recursion on notation ({\sc SRN}) scheme by Bellantoni and Cook~\cite{BC92}. Indeed a {\sc SRN} function can be defined (and typed) by a method \verb!f! in a class $\C$ by:
{\tt
\begin{lstlisting}
int f(int x, int y){
    int res = 0;
    if (x==0) {
        res = g(${\tt y}$);
    } else {
        if (x%2 == 0) {
            res = ${\tt h}_0$(f(x/2,${\tt y}$));
        } else {
            res = ${\tt h}_1$(f(x/2,${\tt y}$));
        }
    }
    return res;
} 
\end{lstlisting}
}
This can be typed provided that ${\tt f} : \C(\tierb) \times \ent(\tierb) \times \ent(\tierb)  \to  \ent(\tiera)$, ${\tt h}_i : \C(\tierb) \times \ent(\tiera) \to \ent(\tiera)$, $i \in \{0,1\}$ and ${\tt g} : \C(\tierb) \times \ent(\tierb) \to  \ent(\tiera)$. As $\tt f$ is recursive, its parameters have to be of tier $\tierb$. However its output can be of tier $\tiera$ provided that we can use only derivations of the shape $\typenv,\Omega \imp_\tierb \ent\ {\tt f}() : \C(\tierb) \times \ent(\tierb) \times \ent(\tierb)  \to  \ent(\tiera)$ and $\typenv,\Omega \imp_\tierb {\tt h}_i : \C(\tierb) \times \ent(\tiera) \to \ent(\tiera) $. 
If $\tt f$ output is of tier $\tiera$ (i.e. computes something) then ${\tt h}_i$ will not be able to recurse on it (as in {\sc SRN}). Clearly, the above program fullfills the general safety criterion for some typing context $\Gamma$ such that $\Gamma(\xa)=\Gamma(\xb)=\ent(\tierb)$ and $\Gamma(\texttt{res})=\ent(\tiera)$ as the recursive calls that are performed in the ${\tt res} = {\tt h}_0({\tt f}(\xa/2,{\tt y}));$ instruction can be given the type $\void(\tierb)$ using Rules \textit{(Ass)} and \textit{(ISub)}.  The operators $\% 2$ and $==0$ can be given the type $\Omega(\% 2) \ni \ent(\tierb) \to \ent(\tierb)$ and $\Omega(==0) \ni \ent(\tierb) \to \bool(\tierb)$ as they are neutral. Finally, the method body is typable using twice the Rule \textit{(If)} on tier $\tierb$ guard and instructions.

\section{Type system non-interference properties}

In this section, we demonstrate that classical non-interference results are obtained through the use of the considered type system. For that purpose, we introduce some intermediate lemmata. Notice that all the results presented in this section hold for compatible pairs only (see Section~\ref{compatible}).

The confinement Lemma expresses the fact that no tier $\tierb$ variables are modified by a command of tier $\tiera$.

\begin{lemma}[Confinement]\label{lem:confinement}
Let $P$ be an AOO program of computational instruction ${\tt Comp}$, $\typenv$ be a contextual typing environment and $\Omega$ be an operator typing environment such that $P$ is (generally) safe with respect to $\typenv$ and $\Omega$.
If $\typenv, \Omega \imp {\tt Comp} :\void(\tiera)$,
then every variable assigned to during the execution of $(\mathcal{I},\underline{{\tt Comp}})$ is of tier $\tiera$.
\end{lemma}

\begin{proof}
By Proposition~\ref{lemma:typepres} and Lemma~\ref{lemma:safeequiv}, we know that $\underline{\texttt{Comp}}$ is safe with respect to environments $\typenv'$ and $\Omega$ such that $\typenv', \Omega \imp \underline{\texttt{Comp}} :\void(\tiera)$.
Now we prove by induction that for any $\mi$ such that $(\mathcal{I},\underline{{\tt Comp}}) \to^{*}(\conf, \mi)$,  every variable assigned to in the evaluation of $\mi$ is of tier $\tiera$:
\begin{itemize}
\item if $ \mi =\epsilon$ or $;$ then the result is straightforward.
\item if $ \mi = [\tau]\ \xa \iasg \mea ;$. Assume $\xa$ is of tier $\tierb$. From rule \textit{(Ass)}, it means that $\alpha'=\tierb$, implying that $\alpha=\tierb$, then that $\mi: \void(\tierb)$.

\item  $\mi = \ca_1 \ \ca_2$ (or  $\iif (\xa)\ithen\{ \ca_1\} \ielse \{\ca_2\}$) then both $\ca_i$ are of tiered type $\void(\tiera)$ by Rule \textit{(Seq)} (respectively \textit{(If)}) and so the result by induction.
\item $ \mi = \xb.\m(\overline{\xa});$ then $\xb.\m(\overline{\xa});$ is of tiered type $\void(\tiera)$ by Rule \textit{(Ass)}. Hence $\m$ cannot be a recursive method because of safety. 
More, $\m$ should be typed as $\typenv',\Omega \imp_\gamma \void \ \m^\C(\overline{\tau}) : \C(\beta) \times \langle \overline{\tau}(\overline{\alpha}) \rangle \to \void(\tiera)$. 
Finally, by Rule \textit{(Body)} the flattened body of $\m$ is also of tier $\void(\tiera)$. 
Consequently, by induction hypothesis, it does not contain assignments of tier $\tierb$ variables. 
\end{itemize}
and so the result.
\end{proof}

\begin{lemma}\label{lem:nowhile0}
Let $P$ be an AOO program of computational instruction ${\tt Comp}$, $\typenv$ be a contextual typing environment and $\Omega$ be an operator typing environment such that $P$ is (generally) safe with respect to $\typenv$ and $\Omega$.
For all $\mi$ such that $(\mathcal{I}, {\tt Comp}) \to^{*} (\conf, \mi)$ and $\typenv, \Omega \imp \mi: \void(\tiera)$, 
there is no while loops and recursive calls evaluated during this execution.
\end{lemma}  
\begin{proof}

From Lemma~\ref{lem:subred}, we know that all instructions reached from a tier~$\tiera$ instruction are of tier~$\tiera$.

\begin{itemize}
\item By rule \textit{(Wh)}, if a while loop occurs, it is of tier $\tierb$.
\item By safety assumption, calls to recursive functions are of tier $\tierb$. \qedhere
\end{itemize}
\end{proof}

We now establish a non-interference Theorem which states that if $\xa$ is variable of tier $\tierb$ then the value stored in $\xa$ is independent from variables of tier $\tiera$. For this, we first need to define the suitable relation on memory configurations and meta-instructions:
\begin{definition}
Let $\typenv$ be a contextual typing environment and $\Omega$ be an operator typing environment. 
\begin{itemize}
\item The equivalence relation $\approx_{\typenv,\Omega}$ on memory configurations is defined as follows: \\
 $\conf \approx_{\typenv,\Omega} \conff$ iff the transitive closure of the pointer graph of $\conf$ corresponding to $\tierb$ variables with respect to $\typenv$ is equal to the one of $\conff$ and both stacks match everywhere except on tier $\tiera$ variables.
\item The relation $\approx_{\typenv,\Omega}$ is extended to commands as follows:
\begin{enumerate}
\item If $\ca=\cb$ then $\ca \approx_{\typenv,\Omega} \cb$
\item If $\typenv, \Omega \imp \ca: \void(\tiera)$ and $\typenv, \Omega \imp \cb:\void(\tiera)$ then $\ca \approx_{\typenv,\Omega} \cb$
\item If $\ca \approx_{\typenv,\Omega} \cb$ and $\cc \approx_{\typenv,\Omega} \cd$ then $\ca\ \cc\approx_{\typenv,\Omega}  \cb\ \cd$ 
\end{enumerate}
\item Finally, it is extended to configurations as follows: \\
 If $\ca \approx_{\typenv,\Omega} \cb$ and $\conf \approx_{\typenv,\Omega} \conff$ then $(\conf,\ca)   \approx_{\typenv,\Omega} (\conff, \cb)$
\end{itemize}
\end{definition}

In other words, two commands that configurations are equivalent for $\approx_{\typenv, \Omega}$ if their memory configurations restricted to tier~$\tierb$ are the same and their commands of tier~$\tierb$ are the same.

\begin{theorem}[Non-interference] \label{thm:soundness}
Assume that $\typenv$ is a contextual typing environment and $\Omega$ is an operator typing environment
such that  $\mi$ and $\mj$ are the computational instructions of two safe programs with respect to $\typenv$ and $\Omega$ having exactly the same classes. 
Assume also that $(\conf, \mi) \approx_{\typenv,\Omega}(\conff,\mj)$. If $(\conf, \mi) \to (\conf',\mi')$ then there exist $\conff'$  and $\mj'$ such that:
\begin{itemize}
\item $(\conff, \mj) \to^*(\conff',\mj')$
\item and $(\conf', \mi') \approx_{\typenv,\Omega}(\conff',\mj')$
\end{itemize}
\end{theorem}

\begin{proof}

We proceed by induction on $\mi$.
\begin{itemize}
\item If $\mi = \xa \iasg \mea ;$. There are two cases to consider. 
\begin{itemize}
\item If $\typenv,\Omega \imp \mi : \void(\tiera)$ then, by definition of $ \approx_{\typenv,\Omega}$, $\typenv,\Omega \imp \mj : \void(\tiera)$. Indeed, either $\mj= \mi$, hence can be typed by tier $\tiera$, or it is of tier $\tiera$ ($\mj$ cannot be a sequence of tier $\tiera$ instructions). Hence Lemma~\ref{lem:confinement} tells us that every variable assigned to in $\mj$ is of tier $\tiera$ and there exists $\conff'$ such that $(\conff,\mj) \to^* (\conff',\epsilon)$ and $(\conf',\epsilon)  \approx_{\typenv,\Omega}(\conff',\epsilon)$. This holds as $\conf'$ and $\conff'$ match on tier $\tierb$ values.
\item If $\typenv,\Omega \imp \mi : \void(\tierb)$ and cannot be given the type $\void(\tiera)$ (i.e. both $\xa$ and $\ea$ are of tier $\tierb$) then, by definition of $\approx_{\typenv,\Omega}$,  $\mi=\mj$. We let the reader check that the evaluation of $\mea$ leads to the same tier $\tierb$ value as it only depends on tier $\tierb$ variables (the only difficulty is for non recursive methods that might have tier $\tiera$ arguments but this is straightforward as it just consists in inlining the method body that is also safe with respect to $\typenv$ and $\Omega$). Consequently, $\conf' \approx_{\typenv,\Omega} \conff'$ as the change on tier $\tierb$
\end{itemize}

\item If $\mi=\mi_1 \ \mi_2$. There are still two cases to consider:
\begin{itemize}
\item Either $\mi$ is of tiered type $\void(\tiera)$ then so is $\mj$ and the result is straightforward.
\item Or $\mi$ is of tiered type $\void(\tierb)$ and $\mj=\mj_1 \ \mj_2$ with $\mi_i \approx_{\typenv,\Omega} \mj_i$. By induction hypothesis if $(\conf,\mi_1) \to (\conf',\mi_1')$ there exists $(\conff',\mj_1')$ such that $(\conff,\mj_1) \to^* (\conff',\mj_1')$ and $(\conf',\mi_1') \approx_{\typenv,\Omega} (\conff',\mj_1')$. Consequently, $(\conf, \mi) \to (\conf', \mi_1' \ \mi_2)$, $(\conff, \mj_1 \ \mj_2) \to^* (\conff', \mj_1' \ \mj_2)$ and $(\conf', \mi_1' \ \mi_2)\approx_{\typenv,\Omega} (\conff', \mj_1' \ \mj_2)$, by definition of $\approx_{\typenv,\Omega}$.
\end{itemize}
\end{itemize}
And so on, for all the other remaining cases.
\end{proof}

Given a configuration $\conf$ and a meta-instruction $\mi$ of a safe program with respect to contextual typing environment $\typenv$ and operator typing environment $\Omega$, the \emph{distinct tier $\tierb$ configuration sequence } $\xi_{\typenv, \Omega}(\conf,\mi)$ is defined by:
\begin{itemize} 
\item If $(\conf,\mi) \to (\conf',\mi')$ then:
$$\xi_{\typenv,\Omega}(\conf,\mi)=
\begin{cases} 
\xi_{\typenv,\Omega}(\conf',\mi')  & \text{if }\conf \approx_{\typenv,\Omega} \conf'\\
\conf.\xi_{\typenv,\Omega}(\conf',\mi') &\text{otherwise}
\end{cases}
$$
\item If $(\conf,\mi) =(\conf,\epsilon)$ then $ \xi_{\typenv,\Omega}(\conf,\mi)= \conf.$
\end{itemize}
Informally, $\xi_{\typenv,\Omega}(\conf,\mi)$ is a record of the distinct tier $\tierb$ memory configurations encountered during the evaluation of $(\conf,\mi)$. Notice that this sequence may be infinite in the case of a non-terminating program.  The relation $\approx_{\typenv,\Omega}$ is extended to sequences by $\epsilon \approx_{\typenv,\Omega} \epsilon$ ($\epsilon$ being the empty sequence)  and $\conf.\xi \approx_{\typenv,\Omega} \conff.\xi'$ iff both $\conf \approx_{\typenv,\Omega} \conff$ and $\xi'\approx_{\typenv,\Omega}  \xi'$ hold.

Now we can show another non-interference property \`a la Volpano et al.~\cite{VIS96} stating that given a safe program, traces (distinct tier $\tierb$ configuration sequence) do not depend on tier $\tiera$ variables.

\begin{lemma}[Trace non-interference]\label{lem:nonin}
Given a meta-instruction $\mi$ of a safe program with respect to environments $\typenv$ and $\Omega$, let $\conf$ and $\conff$ be two memory configurations, if $\conf \approx_{\typenv,\Omega} \conff$ then $ \xi_{\typenv,\Omega}(\conf,\mi) \approx_{\typenv,\Omega}  \xi_{\typenv,\Omega}(\conff,\mi) $. 
\end{lemma}
\begin{proof}
By induction on the reduction $\to$ and a case analysis on meta-instructions $\mi$:
\begin{itemize}
\item If $\mi = \xa \iasg \mea$ then there are two cases to consider:
\begin{itemize}
\item either $\typenv,\Omega \imp \mi : \void(\tiera)$. In this case, by Confinement Lemma~\ref{lem:confinement}, $\xa$ is of tier $\tiera$ and thus if $(\conf,\mi) \to (\conf',\epsilon)$ and $(\conff,\mi) \to (\conff',\epsilon)$, we have $\conf' \approx_{\typenv,\Omega} \conf \approx_{\typenv,\Omega} \conff \approx_{\typenv,\Omega} \conff'$.
\item or  $\typenv,\Omega \imp \mi : \void(\tierb)$ and not $\typenv,\Omega \imp \mi : \void(\tiera)$. In this case, $\xa$ is a tier $\tierb$ variable. However as $\conf   \approx_{\typenv,\Omega} \conff$, $\mea$ evaluates to the same value or reference under both memory configuration. For example, in the case of a variable assignment (Rule (1) of Figure~\ref{fig:sem}), we have $(\conf, [\tau]\ \xa \iasg \xb; ) \to  ( \conf[\xa \mapsto \conf(\xb)], \epsilon )$ and $(\conff, [\tau]\ \xa \iasg \xb; ) \to  ( \conff[\xa \mapsto \conff(\xb)], \epsilon )$. However, as $\conf   \approx_{\typenv,\Omega} \conff$ and $\xb$ is of tier $\tierb$ by Rule \textit{(Ass)} of Figure~\ref{fig:TypeI}, $\conf(\xb) = \conff(\xb)$. Consequently  $ \conf[\xa \mapsto \conf(\xb)]  \approx_{\typenv,\Omega}  \conff[\xa \mapsto \conff(\xb)]$. All the other cases of assignments (operator, methods) can be treated in the same manner, the constructor case being excluded as it enforces the output to be of tier $\tiera$.
\end{itemize}
\item If $\mi = \iwh(\xa)\ido \{ \mi'  \}$ then, by Rule \textit{(Wh)} of Figure~\ref{fig:TypeI}, $\xa$ is enforced to be of tier $\tierb$. Consequently, $\conf(\xa) =\conff(\xa)$ and, consequently, $(\conf, \mi) \to (\conf,\mi')$ and $(\conff, \mi) \to (\conff,\mi')$, for some $\mi'$, independently of whether the guard evaluates to $\true$ or $\false$. As a consequence, $\xi_{\typenv,\Omega}(\conf,\mi')= \xi_{\typenv,\Omega}(\conf,\mi) \approx_{\typenv,\Omega} \xi_{\typenv,\Omega}(\conff,\mi) = \xi_{\typenv,\Omega}(\conff,\mi')$. 
\end{itemize}
All the other cases for meta-instructions can be treated in the same manner and so the result.
\end{proof}

This Lemma implies that tier $\tierb$ variables do not depend on tier $\tiera$ variables.

Using Lemma~\ref{lem:nonin}, if a safe program evaluation encounters twice the same meta-instruction under two configurations equal on tier $\tierb$ variables then the considered meta-instruction does not terminate on both configurations.
\begin{corollary}\label{lem:term}
Given a memory configuration $\conf$ and a meta-instruction $\mi$ of a safe program with respect to environments $\typenv$ and $\Omega$, if $(\conf, \mi) \to^+(\conf',\mi)$ and $\conf \approx_{\typenv,\Omega} \conf' $, then the meta-instruction $\mi$ does not terminate on memory configuration $\conf$.
\end{corollary}

\begin{proof}
Assume that during the transition $(\conf, \mi) \to^+(\conf',\mi)$ there is a
$\conf''$ such that $\conf'' \approx_{\typenv,\Omega} \conf$ does not hold,
then the distinct tier $\tierb$ configuration sequence
$\xi_{\typenv,\Omega}(\conf, \mi)$ contains this $\conf''$. From the construction of the sequence, we
deduce that $\xi_{\typenv,\Omega}(\conf, \mi)$ is of the shape $\ldots
\conf'' \ldots \xi_{\typenv,\Omega}(\conf', \mi)$. However, by Lemma~\ref{lem:nonin}, $\xi_{\typenv,\Omega}(\conf, \mi)  \approx_{\typenv,\Omega} \xi_{\Delta_\tierb}(\conf', \mi)$, hence it is infinite and the
meta-instruction $\mi$ does not terminate on memory configuration $\conf$.

Otherwise, we are in a state $(\conf, \mi)$ from which the set of variables of tier $\tierb$ will never change and containing either a while loop or a recursive call. Consequently, there is some $\conf''$ such that $(\conf', \mi) \to^+(\conf'',\mi)$ and so on. This means that the meta-instruction $\mi$ does not terminate on $\conf$.
\end{proof}

\section{Polynomial time soundness}

It is possible to bound the number of distinct configurations of tier $\tierb$ variables that can be met during the execution of a program, that is the number of different equivalence classes for the $\approx$ relation on configurations.

\begin{lemma}\label{lem:stabledata}\label{lem:polyconf}
Given a safe program with respect to environments $\typenv$ and $\Omega$ of computational instruction ${\tt Comp}$ on input $\mathcal{I}$, the number of equivalence classes for $\approx_{\typenv, \Omega}$ on configurations is at most $|\mathcal{I}|^{n_{\tierb}}$ where $n_{\tierb}$ is the number of tier $\tierb$ variables in the computational instruction.
\end{lemma}

\begin{proof}
First, let us note that the nodes and internal edges of tier $\tierb$ of the pointer graph do not change: new nodes created by constructors can only be of tier $\tiera$ from Rule \textit{(Cons)}. Field assignments can only be of tier $\tiera$ according to Rule $(Ass)$. Moreover, Rule \textit{(Self)} enforces all the field of an object to have tiers matching that of the current object. Consequently, reference type variables of tier~$\tierb$ may only point to nodes of
the initial pointer graph corresponding to input $\mathcal{I}$. The number of such
nodes is bounded by $\size{\mathcal{I}}$ as $\size{\mathcal{I}}$ bounds the number of nodes in the initial pointer graph. 

Second, we look at primitive type variables. By Definition~\ref{typing:op} of operator typing environments, only the output of neutral operators (or of methods of return type of tier $\tierb$) can be applied. Indeed,  they are the only operators to have a tier $\tierb$ output and in an assignment of the shape $\xa \iasg \op(\overline{\xb})$, if $\op$ is of type $\langle \overline{\tau}(\tiera) \rangle \to \tau(\tiera)$ then $\xa$ is enforced to be of tier $\tiera$ by Rule \textit{(Ass)} of Figure~\ref{fig:TypeI}. Neutral operators are operators whose output is either a value of a constant domain ($\bool$, $\charac$, $\ldots$) and, hence, has a constant number of distinct values, or whose output is a positive integer value smaller than one of its input (also tier $\tierb$ values by Definition~\ref{typing:op}). Consequently, they can have a number of distinct values in $O(\size{\mathcal{I}})$, as, by definition, $\size{\mathcal{I}}$ bounds the initial primitive values stored in the initial pointer mapping.

To conclude, let the number of tier $\tierb$
variables be $n_{\tierb}$, the number of distinct possible configurations is 
$\size{\mathcal{I}}^{n_{\tierb}}$.
\end{proof}

Now we can show the soundness: a safe and terminating program terminates in polynomial time.

\begin{theorem}[Polynomial time soundness]\label{sound}
If an AOO program of computational instruction ${\tt Comp}$ is safe with respect to environments $\typenv$ and $\Omega$ and  terminates on input $\mathcal{I}$:
$$(\mathcal{I},\underline{\tt Comp}) \to^k (\conf', \epsilon)$$ then: $$k=O(\size{\mathcal{I}}^{n_\tierb (\nest+ \lev) }).$$
\end{theorem}
\begin{proof}

For simplicity, each recursive call can be simulated through derecursivation  by a while loop instruction of intricacy $\lev$ (keeping the number of reduction steps preserved relatively to some multiplicative and additive constants). Indeed recursive calls do not contain while loops and have only one method call in their body. Consequently, the maximum intricacy of an equivalent program with no recursive call is $\lev + \nest$. 
Now we prove the result by induction on the intricacy $\nest$ of the transformed program:
\begin{itemize}
\item if $\nest=0$ then the program has no while loops. Consequently, $k =O(1)=O(\size{\mathcal{I}}^{n_\tierb\times 0})$
\item if $\nest=k+1$ then, by definition of intricacy, the program contains at least one while loop. Let $\iwh (\xa)\{\mi\}$ be the first outermost while loop of the program. We have that $\nest(\mi)=k$. By Induction hypothesis, $\mi$ can be evaluated in time $O(\size{\mathcal{I}}^{n_\tierb k}$, that is $O(\size{\mathcal{I}}^{n_\tierb (\nest-1)})$. By Lemma~\ref{lem:polyconf},  the tier $\tierb$ variable $\xa$ can take at most $O(\size{\mathcal{I}}^{n_\tierb})$ distinct values. Consequently, by termination assumption, the evaluation of $\iwh (\xa)\{\mi\}$ will take at most $O(\size{\mathcal{I}}^{n_\tierb} \times \size{\mathcal{I}}^{n_\tierb (\nest-1) })$ steps, that is $O(\size{\mathcal{I}}^{n_\tierb \nest })$ steps. Indeed, by Corollary~\ref{lem:term}, we know that a terminating program cannot reach twice the same configuration on tier $\tierb$ variables for a fixed meta-instruction. Now it remains to see that the instruction in the context have only while loops of intricacy smaller than $\nest$. So, by induction again,  they can be evaluated in time $O(\size{\mathcal{I}}^{n_\tierb \nest })$. Finally the total time is the sum so in time $O(\size{\mathcal{I}}^{n_\tierb \nest })$.\qedhere
\end{itemize}
 \end{proof}
 
The same kind of results can be obtained for generally safe programs:
  \begin{proposition}
If an AOO program of computational instruction ${\tt Comp}$ is generally safe with respect to environments $\typenv$ and $\Omega$ and  terminates on input $\mathcal{I}$:
$$(\mathcal{I},\underline{\tt Comp}) \to^k (\conf', \epsilon)$$ then: $$k=O(\size{\mathcal{I}}^{n_\tierb (\nest+ \lev) }).$$
\end{proposition}

As a side effect, we obtain polynomial upper bounds on both the stack size and the heap size of safe terminating programs:

\begin{theorem}[Heap and stack size upper bounds]\label{space}
If an AOO program of computational instruction ${\tt Comp}$ is safe with respect to environments $\typenv$ and $\Omega$ and  terminates on input $\mathcal{I}$ then for each memory configuration $\conf=\langle \heap, \stack_\heap \rangle$ and meta-instruction $\mi$ such that $(\mathcal{I},\underline{\tt Comp}) \to^* (\conf, \mi)$ we have: 
\begin{enumerate}
\item $\size{\heap} = O(\max(\size{\mathcal{I}},\size{\mathcal{I}}^{n_\tierb (\nest+\lev) }))$\label{h}
\item $\size{\stack_\heap}= O( \size{\mathcal{I}}^{n_\tierb (\nest+2\lev) })$\label{s}
\end{enumerate}
\end{theorem}
\begin{proof}
(\ref{h}) By Theorem~\ref{sound}, the number of reductions is in $O(\size{\mathcal{I}}^{n_\tierb (\nest+\lev) })$. The only instructions making the heap (pointer graph) increase are constructor calls. The number of such calls is thus bounded by  $O(\size{\mathcal{I}}^{n_\tierb (\nest+\lev) })$ and consequently the heap size is bounded by the size of the original heap (that is in $O(\size{\mathcal{I}})$) plus the size of the added nodes. Consequently, $\size{\heap}=O(\max(\size{\mathcal{I}},\size{\mathcal{I}}^{n_\tierb (\nest+\lev) }))$ as if $f=O(f')$ and $g=O(g')$ then $f+g=O(\max(f',g'))$. 

(\ref{s}) The number of stack frames added is bounded by $O(\size{\mathcal{I}}^{n_\tierb \lev })$ as, for some recursive method call of signature $s$, each level of recursion may add at most $O(\size{\mathcal{I}}^{n_\tierb})$ stack frames corresponding to method signatures in the set $[s]$. Remember that a stack frame is just a signature of constant size $1$ and a pointer mapping. The size of a pointer mapping is constant on boolean, character and reference type variables (it is equal to $1$) and corresponds to the discrete value stored for numerical primitive type variables. As such values may only augment by a constant for each call to a positive operator and, as such calls may happen $O(\size{\mathcal{I}}^{n_\tierb (\nest+\lev) })$ times, we obtain that each stack frame has size bounded by $O(\size{\mathcal{I}}^{n_\tierb (\nest+\lev) })$. This is also due to the fact that no fresh variable is generated: consequently, the domain of pointer mapping is constantly bounded by the flattened program size, that is linear in the original program size by Corollary~\ref{lem:flatsize}. Here we clearly understand the advantage to deal with flattened programs! All together,  we obtain that $\size{\stack_\heap}= O(\size{\mathcal{I}}^{n_\tierb \lev } \times \size{\mathcal{I}}^{n_\tierb (\nest+\lev) })$, that is $O( \size{\mathcal{I}}^{n_\tierb (\nest+2\lev) })$.
\end{proof}

Again the same results hold for generally safe programs:
\begin{corollary}\label{hs}
If an AOO program of computational instruction ${\tt Comp}$ is generally safe with respect to environments $\typenv$ and $\Omega$ and  terminates on input $\mathcal{I}$ then for each memory configuration $\conf=\langle \heap, \stack_\heap \rangle$ and meta-instruction $\mi$ such that $(\mathcal{I},\underline{\tt Comp}) \to^* (\conf, \mi)$ we have: 
\begin{enumerate}
\item $\size{\heap} = O(\max(\size{\mathcal{I}},\size{\mathcal{I}}^{n_\tierb (\nest+\lev) }))$
\item $\size{\stack_\heap}= O( \size{\mathcal{I}}^{n_\tierb (\nest+2\lev) })$
\end{enumerate}
\end{corollary}

\begin{example}\label{ex:final}
The AOO program of Example~\ref{ex7}. This program is clearly terminating and safe (there is no recursive method call) with respect to the provided environments. Moreover intricacy $\nest$ is equal to $1$ since there is no nested while loops in the method $\tt getTail$. Moreover, its level is equal to $0$ as there is no recursive call.
Moreover there is one tier $\tierb$ variable $n_\tierb=1$. Consequently, applying Theorem~\ref{sound}, we obtain that it terminates in $O(n^1)$, on some input of size $n$. Moreover, by Corollary~\ref{hs}, the heap size and stack size are in $O(\max(n,n^1))=O(n)$ and $O(n^{1(1+2\times 0) })=O(n)$.
\end{example}

\begin{example}\label{cyclic}
Consider the below example representing cyclic data:
{\tt
\begin{lstlisting}
Ring { 
  boolean data;
  Ring next;
  Ring prev;

  Ring(boolean d, Ring old) {
    data = d;
    if (old == null) {
        next = this;
        prev = this;
    } else {
        next = old.next;
        next.setPrev(this);
        prev = old.prev;
        prev.setNext(this);
    }
  }

  boolean getData() {return data;}
  Ring getNext() {return next;}
  Ring getPrev() {return prev;}
  void setPrev(Ring p) {prev = p;}
  void setNext(Ring n) {next = n;}
}

Exe {
  void main() {
    // Init
    Ring a = new Ring(true, null);
    Ring input = new Ring(true, a);
    //Comp: Search for a false in the input.
    copy$^\tierb$ = input$^\tierb$;
    while (copy$^\tierb$.getData() != false) {
        copy$^\tierb$ = copy$^\tierb$.getNext();
    }
  }
}
\end{lstlisting}
}
The program is safe and can be typed with respect to the following judgments:
\begin{itemize}
\item $\verb!getData()! : \verb!Ring!(\tierb) \to \verb!boolean!(\tierb)$
\item $\verb!getNext()! : \verb!Ring!(\tierb) \to \verb!Ring!(\tierb)$
\end{itemize}
with respect to environments $\typenv$ and $\Omega$ such that $\typenv(\verb!copy!)=\typenv(\verb!input!)=\verb!Ring!(\tierb)$. Notice that methods $\verb!setNext(Ring n)!$ and $\verb!setPrev(Ring p)!$ are not required to be typed with respect to tiers as they only appear in the initialization instruction and, consequently, their complexity is not under control (they are just supposed to build the input).
It is obvious that if the main program halts, it will do so in time linear in the size of the input. But it can loop infinitely if the ring does not contain any \verb!false!. We obtain a bound through the use of Theorem~\ref{sound}, that is $O(n^{n_\tierb \times 1})=O(n^2)$. Notice that this bound can be ameliorated to $O(n)$ at the price of a non-uniform formula by noticing that only $1$ tier $\tierb$ variable occurs in the while loop (see the remarks about declassification). Notice that this program could be adapted and typed in while loops of the shape:
\begin{verbatim}
while(copy.getData() != false && n>0){
  ...
  n--;
}
\end{verbatim}
and this would be still typable.

Consequently, we can see that the presented methodology does not only apply to trivial data structures but can take benefit of any complex object structure.
\end{example}

\section{Completeness}
Another direct result is that this characterization is complete with respect to the class of functions computable in polynomial time as a direct consequence of Marion's result~\cite{M11} since both language and type system can be viewed as an extension of the considered imperative language. This means that the type system has a good expressivity. We start to show that any polynomial can be computed by a safe and terminating program. Consider the following method of some class $\C$ computing addition:
{\tt 
\begin{lstlisting}
add(int x, int y){                       
  while(x$^\tierb$>0){
    x$^\tierb$ := x$^\tierb$-1; $:\tierb$
    y$^\tiera$ := y$^\tiera$+1; $:\tiera$
  }
  return y$^\tiera$;
}                             
\end{lstlisting}
}
It can be typed by $\C(\beta) \times \ent(\tierb) \times \ent(\tiera) \to \ent(\tiera)$, for any tier $\beta$, under the typing environment $\Delta$ such that $\Delta(\texttt{add}^\C)(\xa)=\ent(\tierb)$ and $\Delta(\texttt{add}^\C)(\xb)=\ent(\tiera)$. Notice that $\xa$ is enforced to be of tier $\tierb$, by Rules \textit{(Wh)} and \textit{(Op)} as it appears in the guard of a while loop. The operators $>0$ and $-1$ are neutral. Consequently, they can be given the tiered types $\ent(\tierb) \to \bool(\tierb)$ and, $\ent(\tiera) \to \ent(\tiera)$, respectively,  in Rule \textit{(Op)}. The operator $+1$ is positive and its tiered type is enforced to be $\ent(\tiera) \to \ent(\tiera)$ by Rule \textit{(Op)}. 

Consider the below method encoding multiplication:
{\tt
\begin{lstlisting}
mult(int x, int y){
  int z$^\tiera$ := 0;
  while(x$^\tierb$>0){
    x$^\tierb$ := x$^\tierb$-1;
    int u$^\tierb$ := y$^\tierb$;
    while(u$^\tierb$>0){
      u$^\tierb$ := u$^\tierb$-1;
      z$^\tiera$ := z$^\tiera$+1;
    }
  }
  return z$^\tiera$;
}
\end{lstlisting}
}
It can be typed by $\C(\beta) \times \ent(\tierb) \times \ent(\tierb) \to \ent(\tiera)$, for any tier $\beta$, under the typing environment $\Delta$ such that $\Delta(\texttt{mult}^\C)(\xa)=\Delta(\texttt{mult}^\C)(\xb)=\Delta(\texttt{mult}^\C)(\xd)=\ent(\tierb)$ and $\Delta(\texttt{mult}^\C)(\xc)=\ent(\tiera)$. Notice that $\xa$ and $\xd$ are enforced to be of tier $\tierb$, by Rules \textit{(Wh)} and \textit{(Op)}. Moreover $\xb$ is enforced to be of tier $\tierb$, by Rule \textit{(Ass)} applied to instruction $\verb!int u=y;!$, $\xd$ being of tier $\tierb$. Finally, $\xc$ is enforced to be of tier $\tiera$, by Rule \textit{(Op)}, as its stored value increases in the expression $\verb!z+1;!$.

\begin{theorem}[Completeness]\label{complete}
Each function computable in polynomial time by a Turing Machine can be computed by a safe and terminating program.
\end{theorem}
\begin{proof}
We show that every polynomial time function over binary words, encoded using the class $\tt BList$, can be computed by a safe and terminating program. Consider a Turing Machine $TM$, with one tape and one head,  which computes within $n^k$ steps for some constant $k$ and where $n$ is the input size. The tape of $TM$ is represented by two variables $\xa$ and $\xb$ which contain respectively the reversed left side of the tape and the right side of the tape. States are encoded by integer constants and the current state is stored in the variable $\vetat$. We assign to each of these three variables that hold a configuration of TM the tier $\tiera$.
A one step transition is simulated by a finite cascade of if-commands of the form:
{\tt
\begin{lstlisting}[frame=none]
if($\xb{\tt .getHead()}^\tiera$){
	if($\vetat^\tiera==8^\tiera$){
		$\vetat^\tiera$ = $3^\tiera; : \tiera$ 
		$\xa^\tiera$ = $\new\ \tt BList(\false, \xa^\tiera); : \tiera$
		$\xb^\tiera$ = $\xb{\tt .getTail()}^\tiera); : \tiera$
	}else{$\ldots : \tiera$}
}
\end{lstlisting}
}
The above command expresses that if the current read symbol is $\true$ and the state is $8$, then the next state is $3$, the head moves to the right and the read symbol is replaced by $\false$. The methods $\tt getTail()$ and $\tt getHead()$ can be given the types $\tt BList(\tiera) \to BList(\tiera)$ and $\tt BList(\tiera) \to \bool(\tiera)$, respectively (see previous Example).
Since each variable inside the above command is of tier $\tiera$, the tier of the if-command is also $\tiera$.  
As shown above, any polynomial can be computed by a safe and terminating program: we have already provided the programs for addition and multiplication and we let the reader check that it can be generalized to any polynomial. Thus the cascade of one step transitions can be included in an intrication of while loops computing the requested polynomial. Note that this is possible as the cascade will be of tiered type $\void(\tiera)$ and it can be typed by $\void(\tierb)$ through the use of Rule \textit{(ISub)}.
\end{proof}

\begin{corollary}
Each function computable in polynomial time by a Turing Machine can be computed by a generally safe and terminating program.
\end{corollary}
\begin{proof}
By Proposition~\ref{afortiori}.
\end{proof}

\section{Type inference}
Now we show a decidability result for type inference in the case of safe programs.

\begin{proposition}[Type inference]\label{poltypeinf}
Given an AOO program $P$, deciding if there exist a typing environment $\Delta$ and an operator typing environment $\Omega$ such that $P$ is well-typed can be done in time polynomial in the size of the program.
\end{proposition}

\begin{proof}
We work on the flattened program.
Type inference can be reduced to a 2-SAT problem. We encode the tier of each field $\xa$ (respectively instruction $\mi$) within the method $\m$ of class $\C$ by a boolean variable
$x^{\m^\C}$ (respectively $I^{\m^\C}$) that will be true if the variable (instruction) is of tier $\tierb$, false if it is of tier $\tiera$ in  in the context of $\m^\C$. 
Then we generate boolean clauses with respect to the program code:
\begin{itemize}
\item An assignment $\xa := \xb;$ corresponds to $(y^{\m^\C} \wedge x^{\m^\C})$
\item A sequence $\mi_1 \ \mi_2$ corresponds to $(I_1^{\m^\C} \vee I_2^{\m^\C})$
\item An  $\iif (\xa)\ithen\{ \mi_1\}\ielse \{\mi_2\}$ corresponds to $(\neg I_1^{\m^\C} \vee x^{\m^\C}) \wedge (\neg I_2^{\m^\C} \vee x^{\m^\C})$
\item A expression involving a neutral operator $\xa_0\iasg \op(\xa_1,\ldots,\xa_n)$ corresponds to  $\wedge_{i,j} ( x_i^{\m^\C} \vee \neg x_j^{\m^\C})$
\item and so on...
\end{itemize}
The number of generated clauses is polynomial in the size of the flattened program and also of the initial program by Corollary~\ref{lem:flatsize}. The polynomiality comes from the fact that a method may be given distinct types on several calls (thus its body might by typed several times statically). Consequently, tiered types are inferred in polynomial time as 2-SAT problems can be solved in linear time.
\end{proof}

\begin{corollary}
Given an AOO program $P$, deciding if there exist a typing environment $\Delta$ and an operator typing environment $\Omega$ such that $P$ is well-typed and safe can be done in time polynomial in the size of the program.
\end{corollary}
\begin{proof}
Using Lemma~\ref{lemma:safecheck}.
\end{proof}

Just remark that the above corollary becomes false for generally safe programs as a consequence of Proposition~\ref{undec}.

\section{Extensions}\label{sec:ext}
In this section, we discuss several possible improvements of the presented methodology.
\subsection{Control flow alteration}
Constructs altering the control flow like break, return and continue can also be considered in this fragment. For example, a break statement has to be constrained to be of tiered type $\void(\tierb)$ so that if such an instruction is to be executed, then we know that it does not depend on tier $\tiera$ expressions. More precisely, it can be typed by the rule:
$$\ninfer{ }{\typenv,\Omega \imp \breaks; : \void(\tierb)}{(Break)}$$
This prevents the programmer from writing conditionals of the shape:
$$\iwhile(\xa^\tierb)\{\ldots \iif (\xb^\tiera) \ithen \{\breaks;\} \ielse \{\ldots\} \ldots\}$$
 that would break the non-interference result of Lemma~\ref{lem:nonin}. Such a conditional cannot be typed in such a way since $\xb$ has to be of tier $\tierb$ by Rules \textit{(If)} and \textit{(Break)}. Notice that Theorem~\ref{thm:soundness} remains valid since in a \emph{terminating program}, an instruction containing break statements will have an execution time smaller than the same instruction where the break statements have been deleted.
Using the same kind of typing rule, return statements can also be used in a more flexible manner by allowing the execution to leave the current subroutine anywhere in the method body. In previous Section, we have made the choice not to include break statement in the code as this make the definition of a formal semantics much more difficult: at any time, we need to keep in mind the innermost loop executed. The same remarks hold for the continue construct. 

Another possible extension is to consider methods with no restriction on the return statement. 
For the simplicity of the type system, we did not present a flexible treatment of the return statements by allowing the execution to leave the current subroutine anywhere in the method body. But this can be achieved in the same manner. The difficulty here is that all the return constructs have to be of the same tiered type (this is a global check on the method body). Moreover, as usual, we need to enforce that in a conditional all reachable flows lead to some return construct.

\subsection{Static methods, static variables and access modifiers}
The exposed methodology can be extended without any problem to static methods since they can be considered as a particular case of methods. In a similar manner, static variables can be captured since they are global (and can be considered as variables of the executable class) and add no complexity to the program. We claim that the current analysis can handle all usual access modifiers. Indeed the presented work is based on the implicit assumption that all fields are \texttt{private} since there is no field access in the syntax. On the opposite, methods and classes are all \texttt{public}. Consequently, method access restriction only consists in restricting the class of analyzed programs.

\subsection{Abstract classes and interfaces}
Both abstract classes and interfaces can be analyzed by the presented framework. This is straightforward for interfaces since they do not add any complexity to the program. We claim that abstract classes can be analyzed since they are just a particular and simple case of inheritance.
\subsection{Garbage Collecting}

In the pointer graph semantics, dereferenced objects stay in memory forever. It does not entail the bound on the heap size, however, it is far from optimal from a memory usage point of view. However, it is easy to add naive garbage collecting to the system. Indeed, finding which objects in the graph are dereferenced is simply recursively deleting nodes whose indegree is zero (counting pointers in the indegree).

Two different strategies can be thought of to implement dynamically this idea: either use an algorithm similar to the mark-and-don't-sweep algorithm that will color the part of the graph that is referenced then delete the uncolored part; or noting that as assignments and $\pop;$ are the only instructions that may dereference a pointer, maintain for each node of the graph a counter for its indegree, update it at each instruction and delete the dereferenced nodes (and its children if they have no other parents) whenever it happens.

The first strategy has the main drawback that exploring the whole graph will block the execution for some time, especially as starting from each node in each pointer mapping of the pointer stack is needed.

The second strategy will be far more flexible and the real wiping of memory may be deferred if necessary (as the deleting of files in a Unix system for example). An assignment increments the indegree of the node assigned and decrements the indegree of the node previously assigned to the variable. An instantiation ($\new$) adds arrows in the graph, hence increments the indegree of all the nodes associated to its parameter. $\push$ and $\pop$ respectively increment and decrement the indegree of the parameters and the current object of the method. Whenever a node's indegree reaches zero, it can be deleted. Finally, deleting a node of the graph decrements the indegree of the nodes it pointed to.

\subsection{Alleviating the safety condition}

Another way to alleviate the safety condition is to reuse the distinct recursion schemata provided in the {\sc icc} works of the function algebra. For example, over $\verb!Tree!$ structures, a recursive definition -written in a functional manner - of the shape:
{\tt
\begin{lstlisting}
f(t) = new Tree(f(t.getLeft()),f(t.getRight()));
\end{lstlisting}
}
should be accepted by the analysis as recursive calls work on partitioned data and, consequently, the number of recursive calls remains linear in the input size.

One sufficient condition to ensure that property is that:
\begin{itemize}
\item recursive calls of the method body are performed on distinct fields
\item and these fields are never assigned to in the method body 
\end{itemize}
The following depth first tree traversal algorithm satisfies this condition:
{\tt
\begin{lstlisting}
class Tree {
  int node; 
  Tree left; 
  Tree right;
  ...
  void visit() {
    println(node);
    if(left != null){
      left.visit();
    }else{;}
    if(right != null){
      right.visit();
    }else{;}
    
  }
}
\end{lstlisting}
}
on the assumption that the method $\verb!println(n)!$ behaves as expected.
We let the reader check that the method $\verb!visit()!$ can be typed by $\verb!Tree!(\tierb) \to \void(\tierb)$.

\subsection{Declassification}

In non-interference settings, declassification consists in lowering the confidentiality level of part of the data. For example, when verifying a password, the password database itself is at the highest level, knowing whether an input password matches should have the same level of confidentiality. As this is highly impractical, declassifying this partial information makes sense.

In this context, declassifying would mean retyping some tier $\tiera$ variables into $\tierb$. Such a flow is strictly forbidden by the type system, but it would make sense, for example, to compose treatments that are separately well-typed.
As long as those treatments are in finite number, we keep the polynomial bound as bounded composition of polynomials remains polynomial.

Formally, we will say that programs of the form \verb!Exe{void main(){Init I!$_1$ \verb!I!$_2$ ... \verb!I!$_n$\verb!}}! are well-typed iff for each $i\leq n$, \verb!Exe{void main(){Init I!$_1$ ... \verb!I!$_i$\verb!}}! is well-typed when we consider \verb!Init I!$_1$ ... \verb!I!$_{i-1}$ to be the initialization instruction.

\begin{example}
The following program:
{\tt
\begin{lstlisting}
Exe {
  void main() {
    //Init
    int n := ...;
    BList b := null;
    while (n>0) {
      b := new BList(true, b);
      n := n-1;
    }
    //Comp1
    z := 0;
    while (y.getQueue()) {
      z := z+1
    }
    //Comp2
    x := z;
    BList c := null;
    while (x>0) {
      c := new BList(false, c);
    }
    //Comp3
    x := z;
    while (x>0) {
      c := new BList(true, c);
    }
  }
}
\end{lstlisting}
}
cannot be typed without declassification as in \verb!Comp1!, \verb!z! needs to be of tier $\tiera$, while in \verb!Comp2! and \verb!Comp3!, it needs to be of tier $\tierb$.
\end{example}

\begin{remark}
Note that the proof of Theorem~\ref{complete} could be improved by using declassification. Indeed, we could write a first computation instruction that creates a polynomial bound on the number of steps of the Turing Machine, and a second computation instruction that executes the simulation of each step while a counter is lower than the bound. In the first part, the bound needs to be of tier $\tiera$, in the second of tier $\tierb$, hence the use of declassification.
\end{remark}

\section{Related works}
\subsection{Related works on tier-based complexity analysis}
The current work is inspired by three previous works:
\begin{itemize}
\item 
the seminal paper~\cite{M11}, initiating imperative programs type-based complexity analysis using secure information flow and providing a characterization of polynomial time computable functions,
\item
the paper~\cite{HMP13}, extending previous analysis to C processes with a fork/wait mechanism, which provides a characterization of polynomial space computable functions,
\item and the paper~\cite{LM13}, extending this methodology to a graph based language. 
\end{itemize}
The current paper tries to pursue this objective but on a distinct paradigm: Object. It differs from the aforementioned works on the following points:
\begin{itemize}
\item first, it is an extension to the object-oriented paradigm (although imperative features can be dealt with). In particular, it characterizes the complexity of recursive and non-recursive method calls whereas previous works~\cite{M11,HMP13,LM13} where restricted to while loops and to non-object data type (words in ~\cite{M11,HMP13} and records in~\cite{LM13}),
\item second, it studies program intensional properties (like heap and stack) whereas previous papers were focusing on the extensional part (characterizing function spaces). Consequently, it is closer to a programmer's expectations,
\item third, it provides explicit big $O$ polynomial upper bounds while the two aforementioned studies were only certifying algorithms to compute a function belonging to some fixed complexity class.
\item last, from an expressivity point of view, the presented results strictly extend the ones of~\cite{M11}, as the restriction of this paper to primitive data types is mainly the result of~\cite{M11}, while they are applied on a more concrete language than the ones in~\cite{LM13}.
\end{itemize}
The current work is an extended version of~\cite{HP15}. The main distinctions are the following:
\begin{itemize}
\item In~\cite{HP15}, only general safety is studied. Here we have presented a decidable (and thus restricted) safety condition.
\item Contrarily to~\cite{HP15}, the paper presents a formal semantics. The pros are a cleaner theoretical treatment. The cons are that we do not have constructs changing the flow like ``break''. Such constructs can be handled by the tier-based type system. However, their use would increase drastically the complexity of the semantics as a program counter would be required.
\item The requirement that a recursive method can only be called in a tier $\tierb$ instruction is formalized through the use of the $ \imp_\tierb$ notation while it was only informally stated in~\cite{HP15}.
\end{itemize}

\subsection{Other related works on complexity}
There are several related works on the complexity of imperative and object oriented languages. On imperative languages, the papers~\cite{NW06,Moyen09,JonesK09} study theoretically the heap-space complexity of core-languages using type systems based on a matrices calculus. 

On OO programming languages, the papers~\cite{HJ06,HR09} control the heap-space consumption using type systems based on amortized complexity introduced in previous works on functional languages~\cite{HJ03,JostHLH10,BHMS04}. Though similar, the presented result differs on several points with this line of work. First, this analysis is not restricted to linear heap-space upper bounds. Second, it also applies to stack-space upper bounds. Last but not least, this language is not restricted to the expressive power of method calls and includes a \texttt{while} statement, controlling the interlacing of such a purely imperative feature with functional features like recurrence being a very hard task from a complexity perspective. 

Another interesting line of research is based on the analysis of heap-space and time consumption of Java bytecode~\cite{AAGPZ07,AAGPZ12,KSGME12,cachera2005certified}. The results from \cite{AAGPZ07,AAGPZ12,KSGME12} make use of abstract interpretations to infer efficiently symbolic upper bounds on resource consumption of Java programs. A constraint-based static analysis is used in~\cite{cachera2005certified} and focuses on certifying memory bounds for Java Card. The analysis can be seen as a complementary approach since we try to obtain practical upper bounds through a cleaner theoretically oriented treatment. Consequently, this approach allows the programmer to deal with this typing discipline on an abstract OO code very close to the original Java code without considering the corresponding Java bytecode. Moreover, this approach handles very elegantly while loops guarded by a variable of reference type whereas most of the aforementioned studies are based on invariants generation for primitive types only.

The concerns of this study are also related to the ones of~\cite{AGG09,AGG10}, that try to predicts the minimum amount of heap space needed to run the program without exhausting the memory, but the methodology, the code analyzed (source vs compiled) and the goals differ.

A complex type-system that allows the programmer to verify linear properties on heap-space is presented in~\cite{chin2005memory}. The presented result in contrast presents a very simple type system that however guarantees a polynomial bound.

In a similar vein, characterizing complexity classes below polynomial time is studied in \cite{HS09,HofmannS10}. These works rely on a programming language called PURPLE combining imperative statements together with pointers on a fixed graph structure. Although not directly related, the presented type system was inspired by this work.

\subsection{Related works on termination}
The presented work is independent from termination analysis but the main result relies on such analysis. Indeed, the polynomial upper bounds on both the stack and the heap space consumption of a typed program provided by Theorem~\ref{space} only hold for a terminating computation. Consequently, this analysis can be combined with termination analysis in order to certify the upper bounds on any input. Possible candidates for the imperative fragment are \emph{Size Change Termination}~\cite{BA10,BAGM12}, tools like Terminator~\cite{CPR06} based on \emph{Transition predicate abstraction}~\cite{PR05} or symbolic complexity bound generation based on abstract interpretations, see~\cite{G09,GMC09} for example.

\section{Conclusion}

This work presents a simple but highly expressive type-system that is sound and complete with respect to the class of polynomial time computable functions, that can be checked in polynomial time and that provides explicit polynomial upper bounds on the heap size and stack size of an object oriented program allowing (recursive) method calls. As the system is purely static, the bounds are not as tight as may be desirable. It would indeed be possible to refine the framework to obtain a better exponent at the price of a non-uniform formula (for example not considering all tier $\tierb$ variables but only those modified in each while loop or recursive method would reduce the computed complexity. See Example~\ref{cyclic}). OO features, such as abstract classes, interfaces and static fields and methods, were not considered here, but we claim that they can be treated by this analysis. 

This analysis has several advantages:
\begin{itemize}
\item It provides a high-level alternative to usual complexity studies mainly based on the compiled bytecode. The analysis is high-level but the obtained bound are quite tight.
\item It is decidable in polynomial time on the safety criterion.
\item It merges both theoretical and applied results as we both obtain bounds on real programs and a sound and complete characterization of polynomial time computable functions.
\item It is able to deal with the complexity of programs whose loops are guarded by object. This is not a feature of bytecode-based approaches that are restricted to primitive data and are in need of costly program transformation techniques and tools to comply with such kind of programs.
\item It uses previous theoretical techniques (tiering, safe recursion on notation) for functional programs and function algebra and shows that they can be adapted elegantly (though technically) to the OO paradigm.
\item It uses previous security techniques (non-interference, declassification) for imperative programs. The use is slightly different (even orthogonal) but the methodology is surprisingly very close.
\end{itemize}
We expect this paper to be a first step towards the use of tiers and non-interference for controlling the complexity of OO programs. The next steps are the design of a practical application and extensions to (linear or polynomial) space using threads.

\paragraph{Acknowledgments} The authors gratefully acknowledge the advises and comments from anonymous referees that contributed to improving this article.

\bibliographystyle{elsarticle-num}

\end{document}